\documentclass[lettersize,journal]{IEEEtran}
\usepackage{amsmath,amsfonts}
\usepackage{array}
\usepackage[caption=false,font=scriptsize,labelfont=sf,textfont=sf]{subfig}
\usepackage{textcomp}
\usepackage{stfloats}
\usepackage{url}
\usepackage{verbatim}
\usepackage{graphicx}
\usepackage{cite}
\usepackage{xcolor}
\usepackage{amsthm}
\usepackage{amssymb}
\theoremstyle{italic}
\usepackage{amsmath}
\usepackage[table,xcdraw]{xcolor} 
\usepackage{multirow} 
\usepackage{booktabs} 
\usepackage{footnote}
\usepackage{makecell}
\usepackage{tabularx}
\usepackage[ruled,vlined,linesnumbered]{algorithm2e}
\DeclareMathOperator*{\argmax}{arg\,max}
\allowdisplaybreaks[4]
\begin{document}

\title{Joint Optimization for Federated Learning and Transmission over Unreliable Wireless Networks with Heterogeneous Data}

\author{Changheng~Wang,~\IEEEmembership{Student Member,~IEEE,} Xianchao~Zhang,~\IEEEmembership{Member,~IEEE,} Zhiqing~Wei,~\IEEEmembership{Member,~IEEE,} Lingzhu~Zhao, Zhongming~Yang, and Zhiyong~Feng,~\IEEEmembership{Senior Member,~IEEE}
	
\thanks{
	Changheng Wang, Zhiqing Wei, and Zhiyong Feng are with the Key Laboratory of Universal Wireless Communications, Ministry of Education, School of Information and Communication Engineering, Beijing University of Posts and Telecommunications, Beijing 100876, China (e-mail: ch\_wang@bupt.edu.cn; weizhiqing@bupt.edu.cn; fengzy@bupt.edu.cn).
	
	Xianchao Zhang is with the Provincial Key Laboratory of Multimodal
	Perceiving and Intelligent Systems, Jiaxing University, Jiaxing 314001, China (e-mail: zhangxianchao@zjxu.edu.cn).
	
	Lingzhu Zhao is with the School of Electronics and Information, Northwestern Polytechnical University, Xi’an 710000, China (email: zlzzhao@mail.nwpu.edu.cn).
	
	Zhongming Yang is with the School of Information Science and Engineering, Southeast University, Nanjing 210096, China (email: 23023-\linebreak9440@seu.edu.cn).
}}

\markboth{}%
{Shell \MakeLowercase{\textit{et al.}}: A Sample Article Using IEEEtran.cls for IEEE Journals}

\maketitle

\begin{abstract}
In wireless federated learning (FL), data heterogeneity and multiple local updates induce client drift, degrading model convergence. It is further affected by unreliable wireless links, as transmission errors may invalidate model updates. To address these challenges, we propose a federated random walk averaging (FedRW) framework, which is a variant of federated averaging (FedAvg) that mitigates data heterogeneity by updating models along random walk (RW) paths and aggregating them at the server. Model parameters are transmitted in packets with retransmission support to improve training quality by mitigating wireless errors along RW paths. Meanwhile, wireless transmission delays hinder the exploration of FedRW. To this end, we formulate a joint optimization problem that integrates learning, RW path selection, and transmission parameter tuning, aiming to minimize the training loss under delay constraints. By deriving an upper bound on the expected convergence of FedRW over unreliable wireless networks, we reduce the problem to a general form agnostic to task type and model architecture. A distributed solution is then proposed, in which the server or clients optimize packet size and maximum number of retransmissions locally, and efficiently select reliable and expandable next-hop nodes via a resilience-aware beam search with dynamic pruning. Simulation results show that FedRW achieves 2.26\%–9\% higher accuracy than state-of-the-art baselines under high data heterogeneity. Furthermore, the jointly optimized FedRW yields at least 2.78\% higher accuracy and faster convergence compared to baselines.
\end{abstract}


\section{Introduction}
\label{Section1}
\subsection{Motivation}
\IEEEPARstart{F}{ederated} learning (FL) is a distributed training framework that balances computational and communication overhead by performing multiple local training and exchanging model parameters instead of raw data. It supports applications in resource-constrained environments such as unmanned aerial vehicle (UAV) networks, industrial Internet of Things (IIoT), and emergency response networks, encompassing key tasks such as object detection and recognition, equipment monitoring, task allocation and path planning \cite{1Nguyen}.

A key challenge in FL is data heterogeneity, as client data are typically non-independent and identically distributed (Non-IID) due to differences in location, usage, and device configurations, leading to client drift \cite{2Gao}. While multiple local updates reduce communication, they can amplify divergence \cite{3Karimireddy}. In wireless FL, this problem becomes more severe. Since model parameters are transmitted over wireless channels, the instability of wireless links (due to fading, interference, and noise) can cause bit errors, rendering some of model unusable. This weakens the quality of local updates and global aggregation. As a result, data heterogeneity and the unreliability of wireless transmission constitute the two core challenges that significantly impact FL convergence.

In response to Non-IID data distributions, random walk (RW) learning gradually traverses clients via communication links, reducing client drift through distributed updates \cite{4Ayache}. However, its sequential nature limits efficiency. Inspired by \cite{5Wang}, we propose a parallel federated random walk averaging (FedRW) framework. Multiple chains run simultaneously, each traversing clients independently. Upon completing local RW updatas, chains send updated models to the server for global aggregation, without relying on complex aggregation mechanisms. Although this approach introduces additional communication overhead compared to traditional FL, the cost is justified by access to richer and potentially more representative data samples. In essence, it extends the core idea of FL by trading increased computation for reduced communication complexity.

In unreliable wireless networks, parallel RW methods face greater challenges than traditional FL due to their reliance on hop-by-hop model transmission. A single transmission error may cause the entire chain's update to be discarded, leading to wasted resources and slower convergence. Existing mainstream approaches such as \cite{6Chen, 7Zheng} often transmit the entire model as a single packet without supporting retransmissions, which can limit training efficiency in unreliable transmission. To address this, we divide model parameters into multiple packets and introduce a retransmission mechanism. If a packet is in error, only that packet is retransmitted, reducing the risk of updated model loss and lowering communication overhead. To accommodate varying wireless conditions, system error tolerance, and model sizes, transmission parameters such as packet size and maximum retransmissions should be dynamically adjusted. Moreover, since RW path delay and reliability directly affect training efficiency, jointly optimizing path selection and transmission configuration is essential for improving performance in wireless FL settings.

The optimization problem is an integer nonlinear program with high complexity. A centralized solution, where the server handles scheduling, becomes inefficient and impractical for large-scale systems. Therefore, we decompose the problem for local resolution. Each client makes autonomous decisions based on its own state and local link information, enabling a scalable and efficient distributed optimization approach.
\vspace{-1em}
\subsection{Related Work}
We begin by reviewing two key challenges that affect FL convergence: data heterogeneity and deployment over wireless networks, both central to this work. A systematic summary and comparison of the related work is provided in Table \ref{table_1}.

To address data heterogeneity, prior studies have proposed various excellent strategies. Regularization and constraint-based methods \cite{3Karimireddy}, \cite{7Zheng}, \cite{11Li}, \cite{12Sun} aim to reduce client drift by modifying the local objective, supported by strong theoretical convergence guarantees. However, improper tuning may lead to overfitting or underfitting, and constraint-based variants can incur extra computation. Several studies \cite{14Sattler}, \cite{15Ghosh}, \cite{16Ma} propose clustering clients based on local features or model update differences, training specialized models per cluster. Clustered FL demonstrates potential in multi-modal tasks by achieving smaller theoretical generalization error. However, they rely on server-side clustering using limited information, which may yield unstable similarity measures, and suffer from cold-start issues for new clients \cite{17Duan}. Meta-learning methods \cite{18Li}, \cite{19Liu}, \cite{20Chen} aim to build adaptable global models using few-shot local updates, mitigating cold-start and heterogeneity. However, their reliance on Hessian computations and nested structures increases complexity. Semi-federated learning (SFL) \cite{20Ni} adjusts the learning structure to better handle non-IID data, but requires more fine-tuning to effectively balance client-specific updates and global model adaptation. Data-centric approaches \cite{21Yoon}, \cite{22Zhang} attempt to align distributions by sharing or synthesizing auxiliary data. While effective, they compromise privacy and introduce additional bias and overhead \cite{23Li}. RW learning offers an alternative by structuring cross-client update paths, exposing the training process to diverse data distributions without sharing raw data \cite{4Ayache}, \cite{24Triastcyn}, \cite{25Pan}, \cite{27Sun}. This helps reduce convergence bias with moderate communication cost. To improve training efficiency and fault tolerance, recent work explores parallel RW \cite{5Wang}, \cite{28Ye}, where multiple chains aggregate updates periodically or at intersection points to accelerate convergence. However, such decentralized aggregation is inherently complex and lacks global coordination, leading to redundant updates and inter-chain interference that compromise convergence stability.

\begin{table}[!t]
	\centering
	\caption{Technical Comparison of Related Methodologies.}
	\label{table_1}
	\scriptsize
	\renewcommand{\arraystretch}{1.3}
	\setlength{\tabcolsep}{4pt}
	\begin{tabularx}{\columnwidth}{@{} >{\raggedright\arraybackslash}p{2.7cm} >{\raggedright\arraybackslash}p{2.9cm} X @{}}
		\toprule
		\textbf{Methodology} & \textbf{Core Mechanism} & \textbf{Limitation / Gap} \\ 
		\midrule
		\multicolumn{3}{@{} l}{\textit{\textbf{Data Heterogeneity Solutions}}} \\ \addlinespace[1pt]
		Regularization \cite{3Karimireddy,7Zheng,11Li,12Sun} & Local objective constraint & Hyperparameter sensitivity \\ 
		Clustered FL \cite{14Sattler,15Ghosh,16Ma,17Duan} & Similarity-based grouping & Cold-start \& high overhead \\
		Meta-learning \cite{18Li,19Liu,20Chen} & Few-shot adaptability & Hessian-based complexity \\
		SFL \cite{20Ni} & Centralized-Fed hybrid & Paradigm-specific tuning \\
		Data-centric \cite{21Yoon,22Zhang,23Li} & Distribution alignment & Privacy \& bias concerns \\ 
		RW Learning \cite{4Ayache,5Wang,24Triastcyn,25Pan,27Sun,28Ye} & Sequential distribution traversal & Ideal link assumption \\ 
		\midrule
		\multicolumn{3}{@{} l}{\textit{\textbf{Wireless Deployment Optimization}}} \\ \addlinespace[1pt]
		Physical Layer \cite{34Ni,35Wanli} & Signal quality enhancement & Hardware reliance \\
		Resource optimization \cite{6Chen,35Chen,7Zheng,36Liu,37Salari,38Khan,39Ren} & Communication and computing resource & Coarse-grained update \\
		Reliability Control \cite{40Song,41Motamedi,42Razi,43Huang,44Zhao} & Retransmission \& packet sizing & Decoupled from FL theory \\ 
		\midrule
		\textbf{Proposed} \textbf{FedRW} & Joint FL and transmission optimization & Bridging RW topology and wireless reliability \\ 
		\bottomrule
	\end{tabularx}
	\vspace{-1em}
\end{table}

In wireless FL, some excellent studies \cite{30Samarakoon}, \cite{31Alishahi}, \cite{32Elgabli}, \cite{33Hu} have addressed key challenges such as client scheduling and efficient resource allocation, including bandwidth and energy optimization, over reliable wireless networks. However, practical deployments often suffer from unreliable wireless links, leading to model transmission errors \cite{34Ye}, which can disrupt the training process. Given limited wireless resources and client capabilities, such errors are frequently unavoidable. To address this, Ni \textit{et al.} \cite{34Ni} exploited over-the-air computation (AirComp) and reconfigurable intelligent surface (RIS) technologies to enhance the signal transmission quality at the physical layer, effectively minimizing errors in model aggregation. They further extended this approach to a multi-RIS-assisted FL system by jointly optimizing the grouping, phase shift configurations, and transmission power to collaboratively minimize the training loss \cite{35Wanli}. In contrast, another line of research focuses on designing robust resource allocation and intelligent retransmission mechanisms at the algorithmic and network protocol levels, to actively tolerate and address transmission errors. Chen \textit{et al.} \cite{6Chen}, \cite{35Chen} jointly optimized learning, resource allocation, and user selection under packet errors to reduce convergence time and training loss. Zheng \textit{et al.} \cite{7Zheng} enhanced convergence by optimizing power control and client selection. Liu \textit{et al.} \cite{36Liu} proposed a joint communication, sensing, and computing framework to improve data processing in UAV-based FL. Salari \textit{et al.} \cite{37Salari} examined the trade-off between coding rate, convergence time, and accuracy, while Khan \textit{et al.} \cite{38Khan} aimed to mitigate model degradation by reducing packet errors and delays. Ren \textit{et al.} \cite{39Ren} analyzed the joint impact of pruning and transmission errors, providing a closed form solution for optimizing pruning and bandwidth allocation. However, these methods treat each model update as a single packet, which is discarded entirely if errors occur. This severely degrades convergence performance, particularly in FedRW. Introducing fine grained packet segmentation and retransmission mechanisms can mitigate this issue. For example, Song \textit{et al.} \cite{40Song} allowed retransmissions until successful decoding and optimized client participation to minimize average convergence time. Motamedi \textit{et al.} \cite{41Motamedi} designed a two-phase scheme, where retransmissions are triggered based on the correlation between local and global gradients. In addition, several studies have explored adaptive packet sizing strategies. Razi \textit{et al.} \cite{42Razi} adjusted packet size based on channel error rates and communication load  to minimize delay and improve energy efficiency. Huang \textit{et al.} \cite{43Huang} optimized packet sizes according to system state to minimize long-term communication and control cost. Zhao \textit{et al.} \cite{44Zhao} jointly optimized packet size and transmission scheduling to reduce collision probability and delay. Despite these advances, none of these works establish an explicit theoretical link between transmission parameters (e.g., packet size and maximum number of retransmissions) and FL convergence performance, especially in multi-hop structures with spatially correlated wireless links.

\subsection{Contribution and Organization}
This paper proposes a new FL framework to resist data heterogeneity, using RW updates instead of local iterative updates. And jointly consider learning performance, RW path and wireless transmission parameters, enabling the proposed method to operate effectively over unreliable wireless networks. The main contributions are summarized as follows.
\begin{itemize}
	\item \textbf{Wireless FedRW framework for data heterogeneity:} We propose a federated averaging scheme based on parallel RWs, where local models are forwarded along RW chains and updated at each visited client. The final model from each chain is aggregated at the server. Each model update is segmented into multiple packets with retransmission support. Since convergence depends on the quality of RW paths, effective path exploration is crucial.
	\item \textbf{Joint learning and transmission optimization:} Over unreliable wireless networks, transmission parameters affect exploration efficiency and transmission reliability. We formulate a joint optimization problem of RW path selection, packet size, and maximum number of retransmissions to minimize training loss under delay constraints.
	\item \textbf{Convergence guarantee and general loss optimization:} We derive the expected convergence bound of FedRW under packet errors using the Polyak-{\L}ojasiewicz (P{\L}) condition, a widely studied assumption in non-convex optimization. Our analysis establishes an explicit relationship among RW paths, transmission parameters, and learning performance. Based on this, we reformulate the loss minimization into a general optimization problem independent of task type or model structure.
	\item \textbf{Distributed path and transmission coordination:} We develop a distributed solution where the server and each client independently optimize local transmission parameters, and propose a dynamic pruning beam search strategy to efficiently select reliable and expandable next-hop.
\end{itemize}

We evaluate FedRW on MNIST, Fashion-MNIST, CIFAR-10, and CIFAR-100. Under high data heterogeneity, FedRW improves accuracy over FedAvg, Fedprox and FedNova by 2.26\%–9\%. Additional experiments over unreliable wireless networks demonstrate that jointly optimizing RW paths and transmission parameters yields at least 2.78\% accuracy gain and accelerates convergence. We further verified that optimized FedRW is always optimal across diverse system and transmission settings.

The rest of this paper is organized as follows. Section \ref{Section2} introduces the system model. Section \ref{Section3} presents FedRW and the joint optimization problem over unreliable wireless networks. Section \ref{Section4} analyzes convergence and reformulates the problem. Section \ref{Section5} details the distributed solution. Section \ref{Section6} reports experimental results, and Section \ref{Section7} concludes.

\section{System Model}
\label{Section2}
\subsection{Learning Objective}
Consider a wireless network consisting of a set of clients $\mathcal{U} = \{1, 2, \ldots, U\}$ that participate in a distributed learning task, along with a central server. Clients may share their local models with each other or with the server. The server is responsible for collecting local models and coordinating the aggregation of a consensus global model. Suppose that each client $i \in \mathcal{U}$ holds a local dataset of size $b$, denoted by $\mathcal{D}_i=\left\{\xi_{i, \theta}:=(\mathbf{x}_{i, \theta}, y_{i, \theta}) \in \mathbb{R}^d \times \mathbb{R} \text { for } \theta \in[b]\right\}$. The goal of each client $i$ is to minimize the empirical risk $f_i(\mathbf{w}_i)=\frac{1}{b} \sum_{\theta=1}^{b} f(\mathbf{w}_i; \xi_{i, \theta})$ over its local dataset, where $\mathbf{w}_i$ is the local model of client $i$. The global learning objective is \cite{45McMahan}
\begin{equation}
	\label{eq1}
	\min _{\mathbf{w} \in \mathcal{W}} F(\mathbf{w}):=\frac{1}{U} \sum_{i=1}^U f_i\left(\mathbf{w}_i\right),
\end{equation}
where $\mathbf{w} \in \mathcal{W} \subseteq \mathbb{R}^d$ is the global model aggregated by the server, and $\mathcal{W}$ is a closed and bounded feasible set. The optimal global model is defined as $\mathbf{w}^* \in \arg \min _{\mathbf{w} \in \mathcal{W}} F(\mathbf{w})$.

\subsection{Network and Transmission Model}
\renewcommand{\thefootnote}{1}
The wireless network adopts a device-to-device (D2D) communication architecture. Although the server performs the logical role of aggregation, all nodes (including the server and clients\footnote{For clarity, we refer to both the server and clients as nodes.}) are treated equally in terms of link modeling and transmission parameter optimization. Distributed collaboration is achieved through direct inter-node communication, without reliance on centralized coordination. We model the D2D network as a directed graph $\mathcal{G} = (\mathcal{V}, \mathcal{E})$, where $\mathcal{V} = \{1, 2, \ldots, U, U+1\}$ represents the set of $U$ clients and one server. The edge set $\mathcal{E}$ contains directed links, where a directed edge $(i, j) \in \mathcal{E}$ indicates that node $i$ can successfully transmit to node $j$, but not necessarily vice versa, thus capturing potential asymmetries in wireless communication. For each node $i \in \mathcal{V}$, its out-neighbor set $\mathcal{N}_i$ is defined as the set of nodes to which $i$ can directly transmit information.

Each node can independently adjust its transmission parameters, such as packet size and maximum number of retransmissions, based on local information (e.g., link quality). This autonomy enhances communication efficiency and improves overall system robustness under unreliable wireless conditions.

To mitigate mutual interference among the $M$ parallel chains, the system adopts an orthogonal frequency allocation strategy where each chain operates on a dedicated subchannel. Intra-chain collisions are avoided through the sequential nature of model updates within each chain, while inter-chain interference is eliminated through frequency domain isolation. The rate of node $i$ transmitting its FL model parameters to node $j$ is given by \cite{40Song}
\begin{equation}
	\label{eq2}		
		r_{i,j} = 
		B \cdot \mathbb{E}_{h_{i,j}}\left[\log _2\left(1+\frac{P  |{h_{i,j}}|^2 }{B N_0}\right)\right],
\end{equation}
where $h_{i,j} = o_{i,j} d_{i,j}^{-2}$ is the channel gain, $d_{i,j}$ is the distance from node $i$ to $j$, and $o_{i,j} \sim \mathcal{CN}(0,1)$ is the Rayleigh fading parameter with $\mathbb{E}[|o_{i,j}|^2] = 1$. $P$ is the transmission power, $B$ is the bandwidth of the allocated orthogonal subchannel, and $N_0$ is the noise power spectral density.

To enhance reliable model transmission, we divide the model $\mathbf{w}_i$ (or $\mathbf{w}$), of total size $m_0$ bits, into multiple packets. Specifically, $\mathbf{w}_i$ is segmented into $\beta_{i,j} = \lceil \frac{m_0}{\zeta_{i,j}} \rceil$ packets for transmission from node $i$ to node $j$, where each packet carries a payload of $\zeta_{i,j}$ bits and incurs a fixed overhead of $e$ bits. The resulting packet size is $n_{i,j} = \zeta_{i,j} + e$, and the transmission delay per packet is
\begin{equation}
	\label{eq3}
    \delta_{i,j}=\frac{n_{i,j}}{r_{i,j}}.
\end{equation}

Given the unreliability of wireless links, transmitted packets may suffer from bit errors. Assuming a bit error rate (BER) of $\epsilon$, the corresponding packet error rate (PER) can be expressed as \cite{42Razi}
\begin{equation}
	\label{eq4}
	\varrho_{i,j}=1-(1-\epsilon)^{n_{i,j}}.
\end{equation}

To improve transmission reliability, we adopt a packet retransmission scheme. When a packet fails, only that packet is retransmitted, up to a maximum of $R_{i,j}$ attempts. The expected number of transmissions per packet is \cite{42Razi}
\begin{equation}
	\label{eq5}
	 \mathbb{E}\left[R_{i,j}\right] = \sum_{r=1}^{R_{i,j}+1} r \cdot \varrho_{i,j}^{r-1}\left(1-\varrho_{i,j}\right)+\left(R_{i,j}+1\right) \cdot \varrho_{i,j}^{R_{i,j}+1}.	
\end{equation}
Accordingly, the expected delay to transmit the model $\mathbf{w}_i$ over the wireless link from node $i$ to $j$ is given by
\begin{equation}
	\label{eq6}
	\tau_{i,j}=\beta_{i,j} \delta_{i,j} \cdot \mathbb{E}\left[R_{i,j}\right] .	
\end{equation}

This highlights the trade-off between packet size and maximum number of retransmissions. Increasing $n_{i,j}$ reduces the number of packets $\beta_{i,j}$, which helps decrease $\tau_{i,j}$, but also raises $\varrho_{i,j}$. To compensate for the higher PER, a larger $R_{i,j}$ is required, which introduces additional retransmission delays. Therefore, selecting appropriate values for $n_{i,j}$ and $R_{i,j}$ is critical to balancing transmission delay and reliability.

\subsection{Data Heterogeneity}
Client data in FL is often statistically heterogeneous, leading to inconsistencies between local and global objectives. These discrepancies can negatively affect optimization, causing the global model to drift and reducing its generalization performance. To quantify this bias, we adopt the dissimilarity metric proposed in \cite{4Ayache, 11Li}, which measures the gap between local and global objectives.

\newtheorem{definition}{Definition}
\begin{definition}
	\label{definition1}
	A local loss function $f_i$ is said to be $(\alpha, \sigma)$-locally dissimilar at $\mathbf{w}_i$ if
	\begin{equation}
		\left\| \nabla f_i\left(\mathbf{w}_i\right)\right\|^2 \leq \alpha^2+\sigma^2\|\nabla F(\mathbf{w})\|^2,
	\end{equation}
	for $\alpha \geq 0$ and $\sigma \geq 1$. The case $\alpha = 0$, $\sigma = 1$ corresponds to the IID setting. The norm $\|\cdot\|$ denotes the Euclidean norm, and $\nabla f_i(\mathbf{w}_i)$ is the gradient at client $i$.
\end{definition}

\section{Proposed FedRW and Problem Formulation}
\label{Section3}
\subsection{Federated Random Walk Averaging}
\begin{figure}[!t]
	\centering
	\includegraphics[width=3.3in]{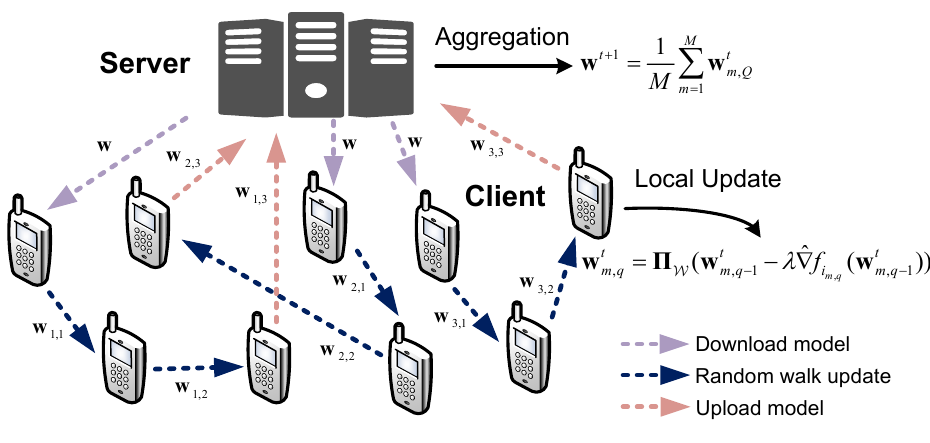}
	\caption{Illustration of the FedRW structure over a reliable wireless network with $M=3$ RW chains and $Q=3$ clients per chain.}
	\label{fig1}
\end{figure}
To attain favorable training outcomes in Non-IID settings, we propose FedRW, a distributed learning algorithm designed to mitigate the impact of data heterogeneity. Unlike FedAvg, which suffers from a communication bottleneck at the server \cite{34Ye}, FedRW replaces local updates with a random walk update scheme. This approach shifts a portion of the communication burden from server-client links to the D2D layer. Although the number of transmission hops increases, FedRW accelerates convergence by distilling and aggregating heterogeneous data features along the RW chains. Consequently, FedRW achieves higher accuracy and lower loss than FedAvg under the same total system communication overhead, as empirically validated in Fig. \ref{fig2}. The training process of FedRW can be formulated as the following optimization problem.
\begin{align}
	\underset{\mathbf{w}_m, m \in \mathcal{M}}{\min}&\frac{1}{M} \sum_{m=1}^M f_m\left(\mathbf{w}_m; \mathcal{D}_m\right) \label{eq8}\\
	\text {s.t.}~~~&~\mathbf{w}_i=\mathbf{w}, ~ i \in \mathcal{U}, \tag{\ref{eq8}{a}} \label{eq8a}
\end{align}
where $\mathcal{M} = \{1, 2, \ldots, M\}$ is the set of RW chains, $\mathcal{D}_m$ is the union of local datasets along the $m$-th chain, and $f_m(\mathbf{w}_m; \mathcal{D}_m)$ denotes the loss of $m$-th RW chain. Constraint (\ref{eq8a}) ensures that all clients share a consistent model upon convergence. To solve (\ref{eq8}), the server sends the global model $\mathbf{w}$ at the beginning of each round to the initial client of each selected RW chain. After each chain completes $Q$ steps, the server collects the updated local models from the final clients and aggregates them to update the global model. The structure of FedRW is illustrated in Fig. \ref{fig1}.

We use stochastic gradient descent (SGD) for RW model updates. In round $t$, the local update of the $q$-th client on the $m$-th chain $i_{m,q}$ is given by
\begin{equation}
	\label{eq9}
	\mathbf{w}_{m,q}^t=\mathbf{\Pi}_\mathcal{W}\left(\mathbf{w}_{m,q-1}^t-\lambda \hat{\nabla} f_{i_{m,q}}\left(\mathbf{w}_{m,q-1}^t\right)\right),
\end{equation}
where $\mathbf{w}_{m,q-1}^t$ is the model parameters received by client $i_{m,q}$ from the previous node in round $t$. $\hat{\nabla} f_{i_{m,q}}$ denotes unbiased gradient estimation of $f_{i_{m,q}}$, and $\lambda$ is the learning rate. $\mathbf{\Pi}_\mathcal{W}$ projects the updated model onto the feasible set $\mathcal{W}$. The client $i_{m,q}$ then transmits the updated model $\mathbf{w}_{m,q}^t$ to its neighbor $i_{m,q+1} \in \mathcal{N}_{m,q}$ for the next local update.

When the server receives updated models from $M$ chains, it performs centralized aggregation using the following scheme to obtain the global model $\mathbf{w}^{t+1}$.
\begin{equation}
	\label{eq10}
	\mathbf{w}^{t+1}=\frac{1}{M} \sum_{m=1}^M \mathbf{w}_{m,Q}^t,
\end{equation}
where $Q$ is the length of the RW chain, and $\mathbf{w}_{m,Q}^t$ represents the local model update from the final client of the $m$-th chain in round $t$. During training, the server and clients collaboratively optimize their models to minimize the loss function in (\ref{eq8}). The complete FedRW training procedure is summarized in Algorithm \ref{algorithm1}.

\begin{algorithm}[t]
	\caption{Federated Random Walk Averaging}
	\label{algorithm1}
	\KwIn{$\lambda$, $M$, $Q$}
	\KwOut{$\mathbf{w}^T$}
	\SetKwInOut{KwIn}{Server executes}
	\KwIn{}
	Initialize global model $\mathbf{w}^0$\;
	\For{$t = 0, 1, 2, \ldots$}{
		Randomly select initial clients for $M$ RW chains and broadcast $\mathbf{w}^t$ to them\;
		\For{$m = 1, 2, \ldots, M$ \textbf{in parallel}}{
			$\mathbf{w}_{m,Q}^t\leftarrow \texttt{RWUpdate} (m, \mathbf{w}^t)$\;
			Aggregate $\mathbf{w}^{t+1}$ via (\ref{eq10}).
		}
	}
	\SetKwInOut{KwIn}{$\texttt{RWUpdate}(m, \mathbf{w}^t)$}
	\KwIn{}
	\For{$q = 1, 2, \ldots, Q$}{
		Client $i_{m,q}$ updates $\mathbf{w}_{m,q}^t$ via (\ref{eq9})\;
		Send $\mathbf{w}_{m,q}^t$ to a random neighbor $i_{m,q+1} \in \mathcal{N}_{m,q}$;
	}
	\Return $\mathbf{w}_{m,Q}^t$ to server\;
\end{algorithm}

Moreover, many techniques developed for FedAvg such as \cite{2Gao}, \cite{3Karimireddy}, \cite{6Chen, 7Zheng, 11Li, 12Sun}, \cite{18Li}, \cite{35Chen}, \cite{49Nguyen} as well as optimization strategies from RW learning, such as client sampling \cite{4Ayache} and adaptive optimizers \cite{24Triastcyn}, \cite{27Sun}, \cite{28Ye} can be naturally extended to FedRW. This flexibility makes FedRW particularly suitable for heterogeneous federated environments. Furthermore, emerging computing paradigms such as quantum computing are expected to further enhance parallelism and potentially enable exponential speedups \cite{46Qiao}.

\subsection{Wireless FedRW Learning Procedure}
During FedRW training process, all models transmit over wireless links. These include the server sending the global model to clients, clients forwarding updates along the RW chains, and sending the final local models back to the server. Due to the unreliability of wireless channels, transmitted models may contain erroneous symbols. Unlike standard FL with local iterations, FedRW involves more frequent model exchanges due to its chain-based structure, making it more vulnerable to transmission errors. Simply discarding a corrupted model can interrupt the training chain, degrade convergence, and waste both computation and communication resources. For instance, if client $i_{m,q}$ fails to deliver its model to $i_{m,q+1}$ due to transmission errors, the chain terminates prematurely, invalidating $q$ local updates and $q+1$ transmissions. To address this, we divide each model into $\beta$ packets and enable packet retransmission, which significantly improves reliability while keeping delays within acceptable limits, thereby preserving training efficiency.

\renewcommand{\thefootnote}{2}
We use the cyclic redundancy check (CRC) to detect errors in model received over the wireless channel and apply automatic repeat request (ARQ) to retransmit any erroneous packets. In the presence of transmission errors, the actual number of successfully received chains at the server may be less than $M$. The global model aggregation (\ref{eq10}) under packet errors can be rewritten as
\begin{equation}
	\label{eq11}
	\mathbf{w}^{t+1}(\mathbf{S}, \mathbf{R}, \mathbf{n})=\frac{\sum\limits_{m=1}^M \prod\limits_{k \in \mathbf{S}_m} C_{m, k} \mathbf{w}_{m,Q}^t}{\sum\limits_{m=1}^M \prod\limits_{k \in \mathbf{S}_m} C_{m, k}},
\end{equation}
where $\mathbf{S}_m$ denotes the path selection of the $m$-th RW chain, which consists of $Q + 1$ wireless links, including the server-to-client link, client-to-client links, and the client-to-server link. The set $\mathbf{S} = [\mathbf{S}_1, \mathbf{S}_2, \ldots, \mathbf{S}_M]$ represents the paths of $M$ RW chains. $\mathbf{R}=[R_{m, k}]^{M \times|\mathbf{S}_m|}$ and $\mathbf{n}=[n_{m, k}]^{M \times|\mathbf{S}_m|}$ are the matrices that record the maximum number of retransmissions and the packet size for each hop in all RW chains, respectively. $C_{m,k} = \{0,1\}$ is an indicator variable that represents whether the model is successfully transmitted over the $k$-th wireless link of the $m$-th chain (with packet size $n_{m,k}$ and up to $R_{m,k}$ retransmissions allowed). Specifically, $C_{m,k} = 1$ indicates a successful transmission, and $C_{m,k} = 0$ indicates a failure. 

The probability of successful transmission for each packet on the $k$-th wireless link of the $m$-th RW chain is \cite{48Chen}
\begin{equation}
	\label{eq12}
	p_{m, k}=1-\left(\varrho_{m,k}\right)^{R_{m, k}+1}.
\end{equation}
Since $\beta_{m, k}$ packets are transmitted independently, $C_{m, k}$ can be defined as
\begin{equation}
	\label{eq13}
	C_{m, k} \sim \textsf{Bernoulli}\left(\left(p_{m, k}\right)^{\beta_{m, k}}\right),
\end{equation}
where the Bernoulli distribution models the probability that $C_{m,k}=1$ as $(p_{m,k})^{\beta_{m,k}}$, meaning all $\beta_{m,k}$ packets are successfully transmitted. The probability that at least one packet fails (i.e., $C_{m,k} = 0$) is $1 - (p_{m,k})^{\beta_{m,k}}$. 

In (\ref{eq11}), the transmission success rate of the entire RW chain $\mathbf{S}_m$ is determined by the product $\prod_{k \in \mathbf{S}_m} C_{m,k}$, which combines the transmission outcomes of all wireless links along the chain. If $\prod_{k \in \mathbf{S}_m} C_{m,k} = 1$, the model updated along the $m$-th RW chain is successfully delivered to the server; otherwise, a failure on any link causes the entire update to be discarded, i.e., $\prod_{k \in \mathbf{S}_m} C_{m,k} = 0$. The sum of these success indicators over all $M$ chains $\sum_{m=1}^M \prod_{k \in \mathbf{S}_m} C_{m,k}$ reflects the transmission reliability of FedRW under the current RW path selection $\mathbf{S}$, packet size $\mathbf{n}$, and maximum number of retransmissions $\mathbf{R}$. The server only aggregates the local models $\mathbf{w}_{m,Q}^t$ from successfully transmitted RW chains.

Additionally, unlike \cite{6Chen, 45McMahan}, our strategy (\ref{eq11}) adopts unweighted aggregation. In the sequential structure of FedRW, weighting a chain’s final model by its total data size, chain length, or unique client count would be inappropriate. This is because client contributions are non-linear and dependent on the position within the chain, as later updates iteratively refine previous ones. The final client of a chain effectively integrates prior updates, inherently reflecting the chain's cumulative contribution. Furthermore, by employing multiple RWs with diverse paths, FedRW naturally achieves broad client coverage and balanced information propagation. Our empirical results confirm that unweighted aggregation provides robust convergence without added complexity or potential biases. Consequently, we maintain unweighted aggregation as an efficient default, leaving the exploration of specialized aggregation strategies for future research.

\subsection{Problem Formulation}
According to (\ref{eq13}), the data packet size and maximum number of retransmissions jointly determine the model transmission success rate along the RW chains, which in turn affects the global model aggregation in (\ref{eq11}). To jointly optimize the transmission parameters and the learning performance of FedRW under wireless networks, we formulate an optimization problem that aims to minimize the training loss. This problem integrates RW path selection with the configuration of packet size and maximum number of retransmissions, achieving a trade-off between algorithm performance and communication efficiency.
\begin{align}
	\min _{\mathbf{S , R , n}} & \frac{1}{M} \sum_{m=1}^M f_m\left(\mathbf{w}\left(\mathbf{S}, \mathbf{R}, \mathbf{n}\right); \mathcal{D}_m\right) \label{eq14}\\
	\text {s.t.}~&~\left|\mathbf{S}_m\right|=Q+1, ~ \forall m \in \mathcal{M}, \tag{\ref{eq14}{a}} \label{eq14a}\\
	&~~n_{m, k} \in \mathbb{N}^{+}, ~ R_{m, k} \in \mathbb{N}, ~\forall m \in \mathcal{M}, ~\forall k \in \mathbf{S}_m, \tag{\ref{eq14}{b}} \label{eq14b}\\
	&~~ \tau_{m, k} \leq \gamma_\tau, ~\forall m \in \mathcal{M}, ~\forall k \in \mathbf{S}_m, \tag{\ref{eq14}{c}} \label{eq14c}\\
	&~~m_0 \leq \Lambda_{m, k} \leq m_0+n_{m, k}, ~\forall m \in \mathcal{M}, ~\forall k \in \mathbf{S}_m, \tag{\ref{eq14}{d}} \label{eq14d}\\
	&~~R_{m, k} \leq \gamma_R, ~ \forall m \in \mathcal{M}, ~\forall k \in \mathbf{S}_m, \tag{\ref{eq14}{e}} \label{eq14e}
\end{align}
where $f_m(\mathbf{w}(\mathbf{S}, \mathbf{R}, \mathbf{n}); \mathcal{D}_m)$ denotes the loss of $m$-th RW chain on the global model $\mathbf{w}$. (\ref{eq14a}) specifies the number of wireless links in each RW chain. (\ref{eq14b}) defines the feasible domains for the packet size and the maximum number of retransmissions. (\ref{eq14c}) imposes a delay limit $\gamma_\tau$ on model transmission over each wireless link. (\ref{eq14d}) ensures that the total number of bits transmitted $\Lambda_{m,k}=n_{m, k} \beta_{m, k}$ is sufficient to cover the entire model size $m_0$ without redundancy. (\ref{eq14e}) limits the maximum number of retransmissions of any single packet by $\gamma_R$.

\section{FedRW Convergence Analysis and Problem Simplification Over Unreliable Wireless Links}
\label{Section4}
Different task types and model architectures can lead to different forms of the loss function $f_m(\mathbf{w}(\mathbf{S}, \mathbf{R}, \mathbf{n}); \mathcal{D}_m)$ in problem (\ref{eq14}). To establish a general formulation, we first analyze the convergence upper bound of FedRW over unreliable wireless networks, in order to understand how the choice of RW paths, packet size, and maximum number of retransmissions affects the convergence performance. Based on this analysis, we then simplify the optimization problem in (\ref{eq14}) using the convergence bound as a surrogate objective. Since the global model update process is influenced by the instantaneous maximum signal-to-noise ratio (SNR), we focus on the expected convergence bound.

\subsection{Convergence Upper Bound Over Unreliable Wireless Networks}
We begin by stating the assumptions required for deriving the expected convergence bound. These assumptions are commonly used in theoretical analyses of FL \cite{12Sun}, \cite{15Ghosh}, \cite{33Hu}. Compared to typical FL analyses, our framework employs weaker and more realistic assumptions. First, it does not rely on the bounded gradient dissimilarity assumption \cite{7Zheng, 16Ma, 39Ren}, allowing for more robust convergence guarantees in highly heterogeneous networks. Second, recognizing that the loss landscape of neural networks is non-convex in general \cite{Li2018Visualizing}, \cite{Sun2020Global}, \cite{Reddi2020Adaptive}, our analysis utilizes the P{\L} condition \cite{50Karimi, 51Sun}. This property is empirically validated in neural network training and enables strong convergence guarantees in non-convex settings. Finally, our analysis dispenses with second-order differentiability or bounded Hessian assumptions \cite{6Chen}, requiring only $L$-smoothness to establish the convergence bound. These relaxed assumptions make our theoretical guarantees more applicable to practical FL involving complex neural network models.

\newtheorem{assumption}{Assumption}
\begin{assumption}
	\label{assumption1}
	The function $F(\mathbf{w})$ is differentiable and satisfies the $\mu$-P{\L} property with constant $\mu > 0$. For any $\mathbf{w} \in \mathcal{W}$, let $F\left(\mathbf{w}^*\right) = \inf_{\mathbf{w} \in \mathcal{W}} F(\mathbf{w})$, it holds that $\frac{1}{2}\left\|\nabla F(\mathbf{w})\right\|^2 \geq \mu \left(F\left(\mathbf{w}\right)-F\left(\mathbf{w}^*\right)\right)$.
\end{assumption}
\begin{assumption}
	\label{assumption2}
	The function $F(\mathbf{w})$ is $L$-smooth with Lipschitz constant $L > 0$. For any $\mathbf{w}, \mathbf{w}' \in \mathcal{W}$, it holds that $\left\|\nabla F(\mathbf{w})-\nabla F\left(\mathbf{w}^{\prime}\right)\right\| \leq L\left\|\mathbf{w}-\mathbf{w}^{\prime}\right\|$.
\end{assumption}

We define the global gradient as
\begin{equation}
	\label{eq18}
	\nabla F\left(\mathbf{w}^t\right)=\frac{1}{\psi} \sum_{m \in \Psi} \nabla f_{i_{m,Q}}\left(\mathbf{w}_{m,Q-1}^t\right),
\end{equation}
where $\Psi$ denotes the set of all $\psi$ possible RW paths, determined by the number of clients $U$ and the chain length $Q$. 

When $U$ and $Q$ are large, the size of $\Psi$ grows exponentially, making direct search for the optimal path set $\mathbf{S}$ computationally intractable. To address this challenge, we decompose the RW path selection into sequential decisions at each node along the chain. Specifically, each decision selects the next node from neighbors, thereby avoiding the infeasibility of solving a large-scale centralized problem and enabling distributed decision-making. Based on (\ref{eq11}), the global model $\mathbf{w}$ at round $t$ is aggregated as
\begin{equation}
	\label{eq19}
	\mathbf{w}^{t+1}=\mathbf{w}^{t}-\lambda\frac{\sum\limits_{m=1}^M \prod\limits_{q=0}^{Q} \sum\limits_{i \in \mathcal{N}_{m, q}} a_{m, q}^i C_{m, q}^i \nabla f_{i_{m,Q}}\left(\mathbf{w}_{m,Q-1}^t\right)}{\sum\limits_{m=1}^M \prod\limits_{q=0}^{Q} \sum\limits_{i \in \mathcal{N}_{m, q}} a_{m, q}^i C_{m, q}^i},
\end{equation}
where the node $i_{m,0}$, $\forall m \in \mathcal{M}$ is the server, and $\mathcal{N}_{m,Q}$ contains only the server. This indicates that in each chain, the model is sent from the server to the client at the first step and returned to the server at the final step. The binary variable $a_{m,q}^i = 1$ indicates that node $i \in \mathcal{N}_{m, q}$ is selected at step $q$ of chain $m$ for the next local update; otherwise, $a_{m,q}^i = 0$. 

Let $\mathbf{w}^*$ denote the optimal global model obtained under an ideal setting with no packet error, using all possible RW paths. Under Assumptions \ref{assumption1}, \ref{assumption2} and Definition \ref{definition1}, we provide the convergence upper bound of FedRW in unreliable wireless networks as stated in the following theorem.

\newtheorem{theorem}{Theorem}
\begin{theorem}
	\label{theorem1}
	Given a set of $M$ random walk paths $\mathbf{S}$, packet size $\mathbf{n}$, maximum number of retransmissions $\mathbf{R}$, and learning rate $\lambda = \frac{1}{L}$, the expected convergence rate is upper bounded by
	\begin{equation}
		\label{eq20}
		\begin{aligned}
			\mathbb{E}&\left[F\left(\mathbf{w}^{t+1}\right)-F\left(\mathbf{w}^*\right)\right] \leq \mathcal{J}^t \mathbb{E}\left[F\left(\mathbf{w}^0\right)-F\left(\mathbf{w}^*\right)\right] \\
			& +\frac{2 \alpha^2}{\psi L }\Bigg(\psi-\sum_{m=1}^M \prod\limits_{q=0}^Q \sum\limits_{i \in \mathcal{N}_{m, q}} a_{m, q}^i E_{m, q}^i\Bigg) \frac{1-\mathcal{J}^t}{1-\mathcal{J}},
		\end{aligned}
	\end{equation}
	where 
	\begin{equation}
		\mathcal{J}=1-\frac{\mu}{L}+\frac{4 \sigma^2 \mu}{\psi L}\Bigg(\psi-\sum_{m=1}^M \prod_{q=0}^Q \sum_{i \in \mathcal{N}_{m, q}} a_{m, q}^i E_{m, q}^i\Bigg),
    \end{equation}
    and $E_{m,q}^i = (p_{m, q}^i)^{\beta_{m, q}^i}$. $\mathbb{E}(\cdot)$ is the expectation over the randomness in packet error.
    
\end{theorem}
\begin{proof}
	Please refer to Appendix \ref{appendixA}.
\end{proof}

\newtheorem{remark}{Remark}
\begin{remark}
	\label{remark1}
	Theorem 1 shows that the expected gap between the global loss $F(\mathbf{w}^{t+1})$ at round $t+1$ and the optimal loss $F(\mathbf{w}^*)$ consists of two parts: 1) an exponentially decaying term related to the initial model error; 2) an error term caused by the unreliability of the wireless links. Since it does not converge when $\mathcal{J} \geq 1$, we focus on the case where $\mathcal{J} < 1$. When $t$ becomes large, the exponential term $\mathcal{J}^t$ tends to zero, and the convergence behavior is dominated by the error term. This term is given by $\frac{2 \alpha^2}{\psi L (1-\mathcal{J})}(\psi-\sum_{m=1}^M \prod_{q=0}^Q \sum_{i \in \mathcal{N}_{m, q}} a_{m, q}^i E_{m, q}^i)$, and depends on the next-hop selection $\mathbf{a}$ and the packet success rates $\mathbf{p}$. Based on (\ref{eq4}) and (\ref{eq12}), for a given $\epsilon$, the $\mathbf{p}$ is determined solely by the packet size $\mathbf{n}$ and the maximum number of retransmissions $\mathbf{R}$. Therefore, optimizing $\mathbf{a}$, $\mathbf{R}$, and $\mathbf{n}$ can effectively reduce the impact of the error term on the convergence upper bound, thus improving the convergence performance of FedRW.
\end{remark}

To ensure convergence in Theorem \ref{theorem1} (i.e., $\mathcal{J} < 1$), the variance $\sigma^2$ must be properly bounded, as follows.

\newtheorem{proposition}{Proposition}
\begin{proposition}
	\label{proposition1}
	Under the same assumptions and settings as Theorem \ref{theorem1}, to ensure the convergence of Theorem \ref{theorem1}, it is necessary to satisfy 
	\begin{equation}
		\label{eq22}
		\begin{aligned}
			1 \leq  \sigma^2<
			\frac{\psi}{\max\limits_{\mathbf{R}, \mathbf{n}} 4\sum\limits_{m \in \psi} \bigg(1- \prod\limits_{q=0}^Q \sum\limits_{i \in \mathcal{N}_{m, q}} a_{m, q}^i E_{m, q}^i\bigg)}.
		\end{aligned}
	\end{equation}
\end{proposition}
\begin{proof}
	For a worst-case estimate of the upper bound on $\sigma^2$, which always guarantees the convergence bound in Theorem \ref{theorem1} under any RW chain, we make all $\psi$ chains participate in the aggregation, i.e., $M = \psi$. Since FedRW converges only when $\mathcal{J}<1$, we have $1-\frac{\mu}{L}+\frac{4 \sigma^2 \mu}{\psi L}\sum_{m \in \psi} (1- \prod_{q=0}^Q \sum_{i \in \mathcal{N}_{m, q}} a_{m, q}^i E_{m, q}^i) < 1$. Since $\mu>0$, and $L>0$, the shift gives $\sigma^2< \frac{\psi}{4\sum_{m \in \psi} (1- \prod_{q=0}^Q \sum_{i \in \mathcal{N}_{m, q}} a_{m, q}^i E_{m, q}^i)}$. To ensure convergence for all choices of $\mathbf{n}$ and $\mathbf{R}$, take $\sigma^2<\frac{\psi}{\max_{\mathbf{R}, \mathbf{n}}4\sum_{m \in \psi} (1- \prod_{q=0}^Q \sum_{i \in \mathcal{N}_{m, q}} a_{m, q}^i E_{m, q}^i)}$. Meanwhile, by Definition \ref{definition1}, $\sigma^2 \geq 1$.
\end{proof}

\begin{remark}
	\label{remark2}
	Proposition \ref{proposition1} reveals a coupling between wireless communication reliability and data heterogeneity. Specifically, as the heterogeneity level $\sigma^2$ increases, the system requires higher model transmission reliability $\sum_{m \in \psi} \prod_{q=0}^Q \sum_{i \in \mathcal{N}_{m, q}} a_{m, q}^i E_{m, q}^i$ to ensure convergence. This highlights that the convergence behavior of FedRW under Non-IID data is highly sensitive to wireless link quality. However, communication reliability cannot be improved simply by reducing the packet size $\mathbf{n}$. Although a smaller $\mathbf{n}$ lowers the PER $\varrho$, it increases the number of packets $\boldsymbol{\beta}$. Moreover, increasing $\mathbf{R}$ can enhance the success probability but at the cost of higher communication overhead and delay. Therefore, achieving robust convergence under high data heterogeneity requires a joint optimization of $\mathbf{n}$, $\mathbf{R}$, and $\boldsymbol{\beta}$, to improve communication quality and enhance the system's tolerance to larger $\sigma^2$.
\end{remark}

Based on Theorem \ref{theorem1}, we derive the convergence bound under ideal conditions.

\newtheorem{lemma}{Lemma}
\begin{lemma}
	\label{lemma_1}
	Under the same assumptions as Theorem \ref{theorem1}, with learning rate $\lambda = \frac{1}{L}$, the convergence upper bound without considering packet errors and RW chain selection is given by
	\begin{equation}
		\mathbb{E}\left[F\left(\mathbf{w}^{t+1}\right)-F\left(\mathbf{w}^*\right)\right] \leq \left(1-\frac{\mu}{L}\right)^t \mathbb{E}\left[F\left(\mathbf{w}^0\right)-F\left(\mathbf{w}^*\right)\right].
	\end{equation}
\end{lemma}
\begin{proof}
	In the ideal case, all $M$ RW chains are selected and transmitted successfully. From (\ref{eq20}), we have $\sum_{m=1}^M \prod_{q=0}^Q \sum_{i \in \mathcal{N}_{m, q}} a_{m, q}^i E_{m, q}^i=M$, and $\psi=M$. Thus, $\psi -\sum_{m=1}^M \prod_{q=0}^Q \sum_{i \in \mathcal{N}_{m, q}} a_{m, q}^i E_{m, q}^i=0$.
\end{proof}

\begin{remark}
\label{remark3}
According to Lemma \ref{lemma_1}, FedRW achieves gap-free convergence to the optimal global model when RW path selection and packet transmission errors are not considered. This result aligns with traditional FL with local updates \cite{6Chen}, since both approaches utilize all available clients in each training round. This consistency is expected, as the idealized FedRW setting essentially reduces to conventional full client participation FL.
\end{remark}

\subsection{Problem Simplification}
Based on Theorem \ref{theorem1}, the impact of wireless factors on the error term of FedRW’s convergence bound can be minimized by optimizing the node selection matrix $\mathbf{a}$ of RW chains, the packet size $\mathbf{n}$, and the maximum number of retransmissions $\mathbf{R}$. By simplifying the error term, we have
\begin{equation}
	\label{eq26}
	\begin{aligned}
		& \frac{2 \alpha^2 }{\psi L (1-\mathcal{J})} \Bigg(\psi-\sum_{m=1}^M \prod_{q=0}^Q \sum_{i \in \mathcal{N}_{m, q}} a_{m, q}^i E_{m, q}^i\Bigg)\\
		& \quad =\frac{2 \alpha^2\left(\psi-\sum_{m=1}^M \prod_{q=0}^Q \sum_{i \in \mathcal{N}_{m, q}} a_{m, q}^i E_{m, q}^i \right)}{\psi \mu-4 \sigma^2 \mu\left(\psi-\sum_{m=1}^M \prod_{q=0}^Q \sum_{i \in \mathcal{N}_{m, q}} a_{m, q}^i E_{m, q}^i \right)}.
	\end{aligned}
\end{equation}

Since $\psi \mu > 0$, $\sigma^2 \mu > 0$, and $\alpha^2 \geq 1$, it is evident that minimizing the error term is equivalent to minimizing $\psi-\sum_{m=1}^M \prod_{q=0}^Q \sum_{i \in \mathcal{N}_{m, q}} a_{m, q}^i E_{m, q}^i$, which in turn corresponds to maximizing $\sum_{m=1}^M \prod_{q=0}^Q \sum_{i \in \mathcal{N}_{m, q}} a_{m, q}^i E_{m, q}^i$. We define $\mathcal{Q} = \{0,1,\dots,Q\}$ and denote the server by $i_s$. Therefore, the original optimization problem in (\ref{eq14}) can be reformulated in the following general form.
\begin{align}
	\max _{\mathbf{a}, \mathbf{R}, \mathbf{n}} & \sum_{m=1}^M \prod\limits_{q=0}^Q \sum\limits_{i \in \mathcal{N}_{m, q}} a_{m, q}^i E_{m, q}^i \label{eq27}\\
	\text {s.t.}&~a_{m, q}^i \in\{0,1\}, ~\forall m \in \mathcal{M}, ~\forall q \in \mathcal{Q}, ~\forall i \in \mathcal{N}_{m, q}, \tag{\ref{eq27}{a}} \label{eq27a}\\
	&~ \sum_{i \in \mathcal{N}_{m, q}} a_{m, q}^i=1, ~\forall m \in \mathcal{M}, ~\forall q \in \mathcal{Q}, \tag{\ref{eq27}{b}} \label{eq27b} \\
	&~ \sum_{m=1}^M \sum_{q=0}^{Q} a_{m, q}^i \leq 1, \forall i \in \mathcal{U}, \tag{\ref{eq27}{c}} \label{eq27c}\\
	&~ i_{m,0} = i_s, ~\mathcal{N}_{m, Q}=\{i_s\}, ~\forall m \in \mathcal{M}, \tag{\ref{eq27}{d}} \label{eq27d}\\
	&~ (\ref{eq14b}) - (\ref{eq14e}), \nonumber
\end{align}
where (\ref{eq27a}) and (\ref{eq27b}) ensure that each RW chain selects exactly one node at each step, and (\ref{eq27c}) ensures that each client is selected at most once across all RW chains. (\ref{eq27d}) guarantees that both the initial and final nodes of each chain correspond to the server $i_s$. The original constraints on $\mathbf{S}$ are equivalently reformulated using the binary decision variables $\mathbf{a}$, enabling a more tractable formulation. The objective of optimization problem (\ref{eq27}) is to maximize the overall wireless link reliability across $M$ chains, thereby reducing the upper bound of the global training loss and accelerating model convergence.

\section{Joint RW Path Selection and Transmission Parameter Optimization}
\label{Section5}
In FedRW, centralized optimization is impractical due to the high communication cost of coordinating global information at the server. Instead, each node must make local decisions independently. To reduce computational complexity and avoid delays caused by slow nodes, we decompose the original problem in (\ref{eq27}) into three distributed subproblems, each solvable locally using only neighbor information.

First, for a fixed maximum number of retransmissions, we determine the optimal packet size for each wireless link in the RW chain. Based on the resulting packet sizes, we compute the corresponding optimal maximum number of retransmissions. Finally, we identify the RW path that minimizes the overall training loss while ensuring reliable communication.

\subsection{Optimal Packet Size}
Since the packet size $n_{m,q}$ only affects the transmission performance of its own wireless link, we can determine it independently. Given a fixed $\mathbf{R} \in \mathbb{N}^{M\times|\mathbf{S}_m|}$ with $\mathbf{R} \leq [\gamma_R]^{M\times|\mathbf{S}_m|}$, the subproblem of selecting the optimal packet size for transmitting the model from node $i_{m,q}$ to a neighbor $i_{m,q+1} \in \mathcal{N}_{m,q}$ can be formulated as
\begin{align}
	\max _{n_{m, q}}& \left(1-\left(1-(1-\epsilon)^{n_{m, q}}\right)^{R+1}\right)^{\beta_{m, q}}  \label{eq28}\\
	\text {s.t.}&~(\ref{eq14c}), (\ref{eq14d}), \nonumber\\
	&~i_{m, q+1} \in \mathcal{N}_{m, q}, ~\forall m \in \mathcal{M}, ~\forall q \in \mathcal{Q},\tag{\ref{eq28}{a}} \label{eq28a}\\
	&~n_{m, q} \leq \gamma_n, ~n_{m, q} \in \mathbb{N}^{+}, ~\forall m \in \mathcal{M}, ~\forall q \in  \mathcal{Q},\tag{\ref{eq28}{b}} \label{eq28b}
\end{align}
where $\gamma_n = \lfloor -\log ( \frac{1 - p_\varrho}{\epsilon} ) \rfloor$ is used to limit the upper bound of PER $p_\varrho$. The delay constraint implicitly enforces a lower bound on packet size to avoid low transmission efficiency. 

Since $\beta$ changes in a stepwise manner as $n$ increases, direct optimization requires handling a discrete and non-smooth objective, which is computationally challenging. To address this, we adopt a logarithmic sampling strategy for $\beta$, as the marginal gain in packet success probability diminishes when packets are short.

For $\lceil \frac{m_0}{\gamma_n} \rceil \leq \beta \leq \lceil \frac{m_0}{n_\tau} \rceil$, where $n_\tau$ is the minimum packet size imposed by the delay constraint, we logarithmically sample $\phi$ values of $\beta$. For each sampled $\beta$, the corresponding optimal packet size $n^*$ is determined using Proposition~\ref{proposition2}, yielding a set of feasible $(n^*, \beta)$ pairs. These pairs are then checked in descending order of $\beta$. Since larger $\beta$ corresponds to smaller $n^*$, if a given $\beta$ satisfies the delay constraint, all smaller $\beta$ will also satisfy it. Finally, we select the feasible $(n^*, \beta^*)$ pair that maximizes the objective.

This approach leverages the fact that the objective is non-monotonic in $n$ because $n$ affects $\beta$. However, once $\beta$ is fixed, the objective increases as $n$ decreases. Therefore, the optimal strategy is to choose the smallest feasible $n$ for each $\beta$ under the delay constraint. The following proposition provides the method to determine the optimal $n$ for a given $\beta$.

\begin{proposition}
	\label{proposition2}
	Given a total model size $m_0$, a fixed number of packets $\beta$, and a maximum number of retransmissions $R$, the optimal packet size $n^*(\beta)$ is given by 
	\begin{equation}
		\label{eq29}
		n^*(\beta)=\max \left\{\min \left\{n_{\Delta}, \gamma_n\right\},\left\lceil\frac{m_0}{\beta}\right\rceil+e\right\},
	\end{equation}
	where $n_{\Delta}$ is the packet size such that $\tau(n_{\Delta}, \beta) =\frac{n_{\Delta} \beta}{r} (\sum_{r=1}^{R+1} r \varrho^{r-1}(n_{\Delta})(1-\epsilon)^{n_{\Delta}}+(R+1) \varrho^{R+1}(n_{\Delta}))$ is closest to $\gamma_\tau$ and satisfies $\tau(n_{\Delta}, \beta) \leq \gamma_\tau$.
\end{proposition}
\begin{proof}	
	For given $m_0$, $R$, and $\beta$, the objective in (\ref{eq28}) is strictly decreasing in $n$, so the smallest feasible $n$ maximizes it. The delay constraint $\tau(n) \leq \gamma_\tau$ imposes a lower bound on $n$, as too small a value would violate it. Thus, $n_{\Delta}$ is the smallest $n$ that strictly satisfies the delay constraint without exceeding $\gamma_n$, ensuring the PER remains acceptable. Additionally, $n$ must satisfy the constraint in (\ref{eq14d}), which requires the payload $\zeta$ of each packet to be no less than $\lceil \frac{m_0}{\beta} \rceil$.
\end{proof}

\begin{remark}
	\label{remark4}
	Proposition \ref{proposition2} shows that $n^*(\beta)$ for a given $\beta$ is determined by $\epsilon$ and model size $m_0$. The objective favors smaller $n^*(\beta)$ to improve transmission success under the delay constraint $\tau(n^*(\beta)) \leq \gamma_\tau$. However, this preference faces two key trade-offs: 1) A larger $m_0$ increases the total bits to transmit. To meet the delay constraint, $n^*(\beta)$ must increase to achieve higher data rates, but this also raises packet failure probability, reducing reliability. 2) A worse channel (e.g., larger $\epsilon$) favors smaller $n^*(\beta)$ for better packet reliability $(1 - \epsilon)^{n^*(\beta)}$, but this leads to more packets and may degrade training performance. Moreover, since $n^*(\beta) \leq \gamma_n$, a small $p_\varrho$ imposes a tighter $n$ bound. If no feasible $n$ satisfies delay constraint, the link is deemed unavailable. While retransmissions can mitigate high error rates, overly conservative $p_\varrho$ values risk unnecessary link exclusion. 
\end{remark}

\subsection{Optimal Maximum Number of Retransmissions}
Based on the optimal $(\mathbf{n}^*, \boldsymbol{\beta}^*)$ pairs from (\ref{eq28}), the subproblem of determining the optimal maximum number of retransmissions for transmitting the model from node $i_{m,q}$ to neighbor $i_{m,q+1} \in \mathcal{N}_{m,q}$ is formulated as
\begin{align}
	\max _{R_{m, q}}& \left(1-\varrho^{R_{m, q}+1}\left(n_{m, q}^*\right)\right)^{\beta_{m, q}^*} \label{eq30}\\
	\text {s.t.}&~(\ref{eq14c}), (\ref{eq14e}), \nonumber\\
	&~ i_{m, q+1} \in \mathcal{N}_{m, q}, ~\forall m \in \mathcal{M}, ~\forall q \in \mathcal{Q}, \tag{\ref{eq30}{a}} \label{eq30a} \\
	&~ R_{m, q} \in \mathbb{N}, ~\forall m \in \mathcal{M}, ~\forall q \in \mathcal{Q}. \tag{\ref{eq30}{b}} \label{eq30b}
\end{align}

We determine the optimal maximum number of retransmissions for per wireless link using the following proposition.
\begin{proposition}
	\label{proposition3}
	Given the optimal $n^*$ and $\beta^*$, the optimal maximum number of retransmissions $R^*$ is given by
	\begin{equation}
		\label{eq31}
		R^*=\min \left\{R_{\Delta}, \gamma_R\right\},
	\end{equation}
	where $R_{\Delta}$ is the largest $R$ satisfying $\tau(R_{\Delta}) \leq \gamma_\tau$ and closest to $\gamma_\tau$, with $\tau(R_{\Delta}) =\frac{n^* \beta^*}{r} (\sum_{r=1}^{R_{\Delta}+1} r \varrho^{r-1}(n^*)(1-\epsilon)^{n^*}+(R_{\Delta}+1) \varrho^{R_{\Delta}+1}(n^*))$.
\end{proposition}
\begin{proof}
	Since $\epsilon \in (0,1)$, the objective in (\ref{eq30}) increases monotonically with $R$. Thus, the largest $R$ that satisfies $\tau(R) \leq \gamma_\tau$ maximizes the objective while respecting the delay constraint and the bound in (\ref{eq14e}).
\end{proof}

\begin{remark}
	\label{remark5}
	Proposition \ref{proposition3} shows that the optimal $R^*$ depends on both $\epsilon$ and the amount of transmitted data $n^*\beta^*$. A larger $n^*\beta^*$ tightens the delay constraint, potentially requiring a smaller $R^*$, which increases the risk of transmission failure. Conversely, a larger $\epsilon$ requires more retransmissions to offset higher packet error. Therefore, selecting optimal $(n^*, \beta^*)$ pair and $R^*$ reflects a fundamental trade-off among transmission efficiency, channel reliability, and training loss over unreliable wireless networks.
\end{remark}

\subsection{Optimal RW Path Selection}
Based on the optimal $\mathbf{n}^*$, $\boldsymbol{\beta}^*$ and $\mathbf{R}^*$, the optimization problem in (\ref{eq27}) can be simplified as
\begin{align}
	\max _{\mathbf{a}} & \sum_{m=1}^M \prod_{q=0}^Q \sum_{i \in \mathcal{N}_{m, q}} a_{m, q}^i E_{m, q}^i \label{eq33}\\
	\text {s.t.}&~(\ref{eq27a}) - (\ref{eq27d}). \nonumber
\end{align}

Formally, this problem resembles a maximum weight bipartite matching or a minimum cost flow problem. However, the objective involves nonlinear products across multiple hops, making matching or flow-based solutions inapplicable. Moreover, in sparse wireless topologies with local link constraints, centralized coordination would incur excessive communication overhead and delay. Therefore, we develop a more scalable distributed solution.

To enable efficient node selection in wireless FedRW, we propose a resilience-aware beam search with dynamic pruning strategy. This strategy considers both current link reliability and the extendibility of the path. For node $i_{m,q}$, the set of next-hop nodes is defined as
\begin{equation}
	\label{eq34}
	\mathcal{Z}_{m, q+1}=\left\{i_{m, q+1} \in \chi_{m, q} \mid i_{m, q+1} \notin \mathcal{A}, \chi_{m, q} \subseteq \mathcal{N}_{m, q}\right\},
\end{equation}
where $\chi_{m, q}$ denotes the set of neighbors filtered by (\ref{eq28}) and (\ref{eq30}) to exclude nodes that cannot satisfy the delay constraint even with optimized transmission parameters. $\mathcal{A}$ denotes the set of clients already selected by RW chains in the current round to satisfy constraint (\ref{eq27c}).

Given a beam size $\rho$, we select the top-$\rho$ candidate nodes $X_{m, q+1}=\left\{i_{m, q+1}^1, i_{m, q+1}^2, \ldots, i_{m, q+1}^\rho\right\}$ from $\mathcal{Z}_{m,q+1}$ with the highest link reliability, where the nodes are sorted in descending order of reliability. The core idea of dynamic pruning is to restrict beam expansion to only the most promising candidates based on current network conditions. The next-hop node selection rule is 
\begin{equation}
	\label{eq35}
	i_{m,q+1}=\left\{\begin{array}{ll}
		i_{m, q+1}^{\xi^*}, \quad \text {if}~ \exists \xi^*=\min \left\{\xi : \mathcal{N}_{m, q+1}^{\xi}=\rho\right\}, \\
		i_{m, q+1}^{\xi^{\dagger}}, \quad \text {if}~ \nexists \xi^*, ~ \exists \xi^{\dagger}=\argmax\limits_{\xi : \mathcal{N}_{m, q+1}^{\xi}>0} \mathcal{N}_{m, q+1}^{\xi}, \\
		\emptyset, \quad \text{otherwise},
	\end{array}\right.
\end{equation}
where $\xi=\{1,2,\ldots,\rho\}$, and $\mathcal{N}_{m, q+1}^{\xi}$ is the number of available neighbors of $i_{m, q+1}^{\xi}$ for the $\xi$-th node in $X_{m, q+1}$.

The selection function (\ref{eq35}) first searches the candidate set $X_{m,q+1}$ in descending order of link reliability to find the most reliable node that can expand to $\rho$ neighbors. If none is found, it selects the one with the largest number of available neighbors. If no candidates have any valid next-hop neighbors, the path is terminated. Once a node is selected, all other candidates in the same beam are pruned to reduce branching overhead.

While the selection in (\ref{eq35}) is deterministic per decision instant, the process remains a random walk as the neighbor set and link reliability are stochastic observations of random topologies and fading channels. By adapting to these dynamic inputs, the algorithm generates a path that is a stochastic realization of the evolving network state.

Algorithm \ref{algorithm2} summarizes the joint RW path selection and transmission parameter optimization algorithm over unreliable wireless networks.
\begin{algorithm}[t]
	\caption{Proposed FedRW Over Unreliable Wireless Networks}
	\label{algorithm2}
	\KwIn{$\lambda$, $M$, $Q$, $\gamma_\tau$, $\gamma_R$, $p_\varrho$, $\phi$, $\rho$}
	\KwOut{$\mathbf{w}^T$}
	\SetKwInOut{KwIn}{Server executes}
	\KwIn{}
	Initialize global model $\mathbf{w}^0$\;
	\For{$t = 0, 1, 2, \ldots$}{
		Server calculates optimal $n^*$, $\beta^*$ and $R^*$ for each of its neighbors\;
		Select $M$ initial clients by (\ref{eq35}) and broadcast $\mathbf{w}^t$ to them\;
		\For{$m = 1, 2, \ldots, M$ \textbf{in parallel}}{
			$\mathbf{w}_{m,Q}^t\leftarrow \texttt{RWUpdateWirel} (m, \mathbf{w}^t)$\;
			Aggregate $\mathbf{w}^{t+1}$ via (\ref{eq11}).
		}
	}
	\SetKwInOut{KwIn}{$\texttt{RWUpdateWirel}(m, \mathbf{w}^t)$}
	\KwIn{}
	\For{$q = 1, 2, \ldots, Q$}{
		Client $i_{m,q}$ updates $\mathbf{w}_{m,q}^t$ via (\ref{eq9}) and calculates optimal $n^*$, $\beta^*$ and $R^*$ for each of its neighbors\;
		Select $i_{m,q+1}$ by (\ref{eq35}) and transmit $\mathbf{w}_{m,q}^t$ to it;
	}
	\Return $\mathbf{w}_{m,Q}^t$ to server\;
\end{algorithm}

\subsection{Implementation and Complexity}
In wireless FedRW, each node estimates the SNR to its neighbors and computes the transmission rate. Given the delay bound $\gamma_\tau$, retransmission limit $\gamma_R$, and packet error constraint $p_\varrho$, the optimal transmission parameters $(\mathbf{n}^*, \boldsymbol{\beta}^*)$ and $\mathbf{R}^*$ are determined. Wireless links that violate the delay constraint are removed, resulting in the feasible neighbor set $\chi$. The server maintains a global set $\mathcal{A}$ to track selected clients. After each local update, the current node sends the next client identifier to the server. If no direct link exists, this information is relayed through reverse chaining. For next-hop selection, each node uses $\chi$ and $\mathcal{A}$ to construct the candidate set $\mathcal{Z}$ and selects the top-$\rho$ nodes $X$ with the highest link reliability. Redundant branches are pruned based on condition (\ref{eq35}).

According to Proposition 2, each node performs a binary search over $n \in [n_\tau, \gamma_n]$. Each step has complexity $\mathcal{O}(R)$, and a total of $\mathcal{O}(\log_2 (\gamma_n-n_\tau))$ iterations are required. This is repeated $\phi$ times for feasible $\beta$, resulting in a total complexity of $\mathcal{O}(\phi \cdot R \cdot \bar{d} \cdot \log_2 (\gamma_n-n_\tau))$, where $\bar{d}$ is the average node degree.
Similarly, Proposition 3 uses binary search over $R \in [0, \gamma_R]$ with complexity $\mathcal{O}(\gamma_R \cdot \bar{d} \cdot \log_2 \gamma_R)$. In the resilience-aware beam search, each node selects the top-$\rho$ candidates from $\mathcal{Z}$ with complexity $\mathcal{O}(z \log z)$, where $z = |\mathcal{Z}|$. Evaluating the extendibility of each candidate costs $\mathcal{O}(\bar{d})$, and pruning stops early once the first expandable candidate is found. Thus, on average, only $\rho_p \ll \rho$ candidates are evaluated. The total complexity is $\mathcal{O}(z \log z + \rho_p \bar{d})$, which is near-linear in practice as $\rho_p$ is typically small.

The modular design enables each component to be optimized independently, facilitating future extensions such as learning-based path selection or adaptive transmission under dynamic networks, and the integration of quantum algorithms to improve both computational and communication efficiency.

FedRW mitigates the impact of channel state information (CSI) staleness through a proactive local sensing mechanism. Under a block fading channel, the current node performs an instantaneous SNR probe for all neighbors before each transition. Based on these measurements, the node subsequently computes the optimal transmission parameters and reliability to execute the beam search selection. Because the computational complexity is near-linear, the delay between sensing and transmission is negligible compared to the channel coherence time. This ensures that the selected parameters remain valid throughout the duration of each hop.

\subsection{Optimality of Problem Decomposition}
To rigorously justify the proposed decomposition, we provide a proof of equivalence between the joint optimization problem and the distributed subproblems.

\begin{proposition}
	\label{proposition_equivalence}
	The joint optimization problem in (\ref{eq27}) is equivalent to maximizing the reliability $E_{m,q}^i$ for each potential wireless link independently, followed by solving the node selection matrix $\mathbf{a}$ based on these maximized values.
\end{proposition}

\begin{proof}
	Since the constraints (\ref{eq14c})$-$(\ref{eq14e}) are link-specific, the feasible region for each $(n_{m,q}, R_{m,q})$ is strictly local. Importantly, there are no coupling constraints that link the parameter choices of different wireless links. Thus, the global joint maximization can be nested as
	\begin{equation}
		\label{eq_proof_nested}
		\max_{\mathbf{a}} \left( \sum_{m=1}^M \prod_{q=0}^Q \sum_{i \in \mathcal{N}_{m,q}} a_{m,q}^i \left[ \max_{n_{m,q}, R_{m,q}} E_{m,q}^i \right] \right).
	\end{equation}
	The $\max_{n_{m,q}, R_{m,q}}$ operator can be moved inside the summations and products because the objective is monotonically non-decreasing with respect to each $E_{m,q}^i \geq 0$ and the local feasible sets are independent. By optimizing each wireless link reliability, the problem reduces to (\ref{eq33}), confirming the optimality of the decomposition.
\end{proof}

\begin{remark}
	\label{remark_efficiency}
	We adopt two practical approximations for real-time execution. First, the two dimensional search is decoupled into two sequential one-dimensional subproblems. A logarithmic sampling strategy for $\beta$ is used to prune the search space while capturing diminishing marginal gains of transmission efficiency. Second, a resilience-aware beam search with an extendibility criterion (\ref{eq35}) is employed to avoid exponential complexity in path selection. This strategy prunes suboptimal branches early to achieve near-optimal reliability with linear complexity. Numerical results demonstrate that the cumulative optimality gap is negligible.
\end{remark}

\section{Simulation Results and Analysis}
\label{Section6}
Our objective is to validate the capability of FedRW to mitigate data heterogeneity, as well as the superiority of the proposed joint RW path selection and transmission parameter optimization algorithm. Specifically, we consider a simulation environment consisting of 100 clients and a central server. The server is positioned at the center of a circular area with a radius of 500 m, and clients are uniformly distributed within this area. We construct a directed Bernoulli random graph by connecting each pair of nodes independently with a probability of 0.5. The simulation results are the average of multiple experiments. To demonstrate the robustness, the figures include shaded regions representing the standard deviation.

We evaluate FedRW on four benchmark datasets. For both MNIST and Fashion-MNIST, we use a multilayer perceptron (MLP) (784$\times$100$\times$10) with ReLU activation and Softmax output, trained with cross-entropy loss. For CIFAR-10 and CIFAR-100, we adopt the VGG13 \cite{47Simonyan} and ResNet-18 \cite{52He} respectively, both with Softmax and cross-entropy loss. Simulation parameters are summarized in Table \ref{table2}, which serves as the default unless otherwise stated.

\begin{table}[!t]
	\centering
	\caption{Simulation Parameters}
	\label{table2}
	\begin{tabular}{|c|c|c|c|}
		\hline
		\multicolumn{2}{|c|}{\textbf{Wireless Parameters}} & \multicolumn{2}{|c|}{\textbf{Algorithmic Parameters}} \\
		\hline
		\textbf{Parameter} & \textbf{Value} & \textbf{Parameter} & \textbf{Value} \\
		\hline
		$B$ & 10 MHz & $U$ & 100 \\
		\hline
		$P$ & 23 dBm & $M$ & 5 \\
		\hline
		$N_0$ & $-$174 $\mathrm{dBm/Hz}$ & $Q$ & 3 \\
		\hline
		$\epsilon$ & 1e-5 & Batch size & 64 \\
		\hline
		$\gamma_{\tau}$ & 20 ms & $\lambda$ & 0.01 \\
		\hline
		$p_\varrho$ & 0.1 & $\phi$ & 100 \\
		\hline
		$\gamma_{R}$ & 2 & $\rho$ & 4 \\
		\hline
		$e$ & 128 bits &  & \\
		\hline
	\end{tabular}
\end{table}

We consider two data heterogeneity settings. 1) In the mixed Non-IID setting, each client has a $\zeta_n$ fraction of IID data and a $(1-\zeta_n)$ fraction from a label‑sharded Non‑IID pool (one class per shard) \cite{45McMahan}, where smaller $\zeta_n$ increases heterogeneity. 2) In the Dirichlet Non-IID setting, samples of each class $y$ are allocated to clients by sampling a probability vector $\mathbf{p}_y$ from a Dirichlet distribution with concentration parameter $\zeta_d$, i.e., $\mathbf{p}_y \sim \mathrm{Dir}(\zeta_d)$ \cite{5Wang}. A smaller $\zeta_d$ yields more skewed class distributions across clients.

Experiments are conducted on a computing node equipped with an Intel Xeon Platinum 8558 CPU and an NVIDIA RTX 4090 D GPU. The implementation uses Python 3.8 and PyTorch 1.13.0 with CUDA 11.8 for GPU acceleration.

\subsection{Performance Comparisons of FedRW Learning Framework}
We first evaluate FedRW in heterogeneous networks to demonstrate its effectiveness in handling non-IID data. As shown in Fig. \ref{fig2}, the X-axis indicates the cumulative communication overhead across all nodes over training. Under highly heterogeneous data distributions on MNIST, FedRW achieves higher accuracy and faster convergence than FedAvg under the same communication overhead, indicating that FedRW adheres to the core principle of FL, trading computation for communication efficiency. In particular, FedRW improves the final accuracy by 8.72\% and 2.26\%, respectively. Moreover, because FedRW traverses a more diverse set of data samples, it exhibits more stable convergence behavior. 

\begin{figure}[!t]
	\centering	
		\subfloat[MNIST (mixed Non-IID, $\zeta_n$=0)]{
			\includegraphics[width=0.48\linewidth]{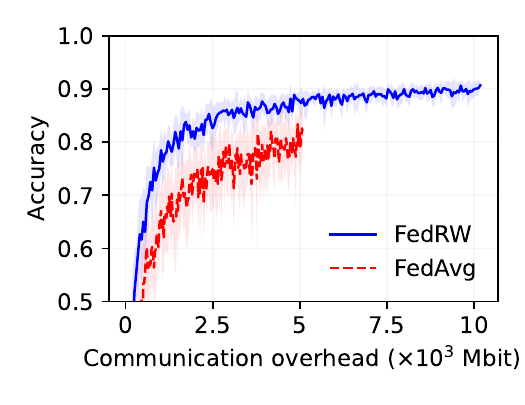}
		}
		\subfloat[MNIST (Dirichlet Non-IID, $\zeta_d$=0.1)]{
			\includegraphics[width=0.48\linewidth]{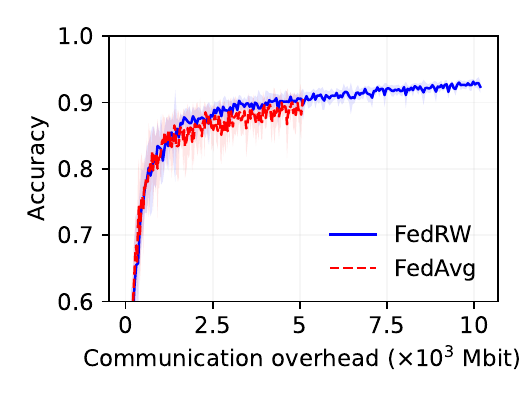}
		}
		\caption{MNIST classification accuracy vs. total communication overhead under Non-IID data distribution (200 rounds).}
		\label{fig2}
\end{figure}

\begin{figure}[!t]
	\centering	
		\subfloat[Fashion-MNIST (mixed Non-IID, $\zeta_n$=0)]{
			\includegraphics[width=0.48\linewidth]{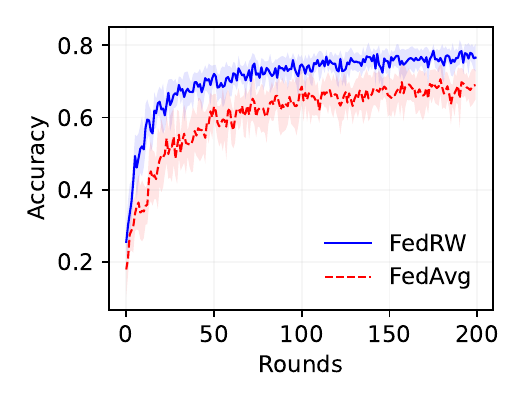}
		}
		\subfloat[Fashion-MNIST (mixed Non-IID, $\zeta_n$=0)]{
			\includegraphics[width=0.48\linewidth]{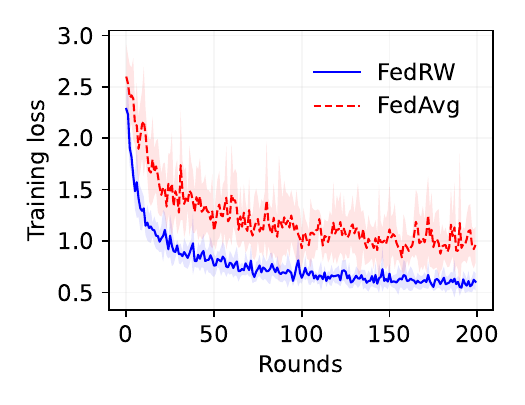}
	}
	
	\vspace{-0.4cm} 
	
		\subfloat[CIFAR-10 (mixed Non-IID, $\zeta_n$=0.2)]{
			\includegraphics[width=0.48\linewidth]{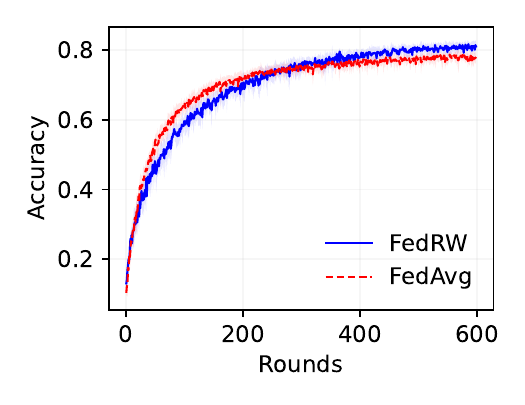}
		}
		\subfloat[CIFAR-10 (mixed Non-IID, $\zeta_n$=0.2)]{
			\includegraphics[width=0.48\linewidth]{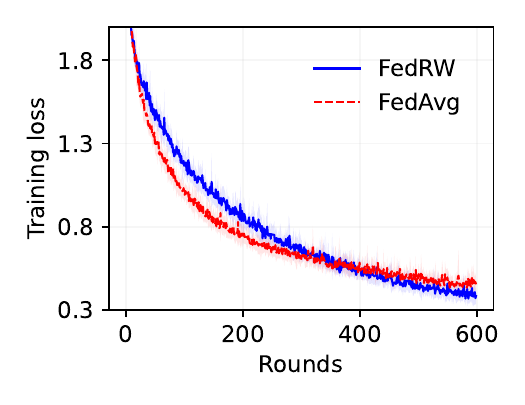}
	}
	
	\vspace{-0.4cm} 
	
		\subfloat[CIFAR-100 (mixed Non-IID, $\zeta_n$=0.4)]{
			\includegraphics[width=0.48\linewidth]{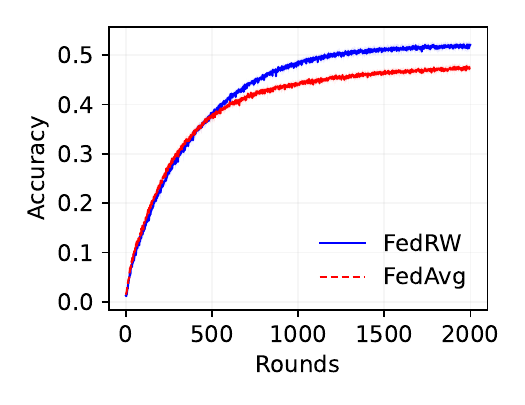}
		}
		\subfloat[CIFAR-100 (mixed Non-IID, $\zeta_n$=0.4)]{
			\includegraphics[width=0.48\linewidth]{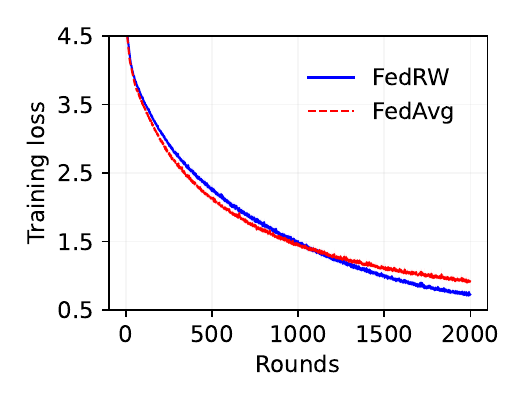}
		}
		\caption{Classification accuracy and training loss on Fashion-MNIST, CIFAR-10, and CIFAR-100 in mixed Non-IID settings.}
		\label{fig3}
\end{figure}

Fig. \ref{fig3} shows the accuracy and training loss of algorithms trained in strongly Non-IID settings on additional datasets, with $\gamma_\tau=$ 3 s on CIFAR-10 and CIFAR-100. For Fashion-MNIST, the results are similar to previous findings, with FedRW achieving 9\% higher accuracy and faster reduction in training loss than FedAvg. However, for CIFAR-10 and CIFAR-100, FedRW explores the parameter space with higher variance updates in the early stages, leading to slower convergence initially. In contrast, FedAvg tends to converge quickly to a suboptimal solution dominated by local data distributions, which appears more favorable during the early rounds. As training progresses, FedRW achieves lower training loss, eventually improving accuracy by 2.49\% and 4.97\% respectively, as it better captures the global data distribution.

\begin{figure}[!t]
	\centering
	\includegraphics[width=0.48\linewidth]{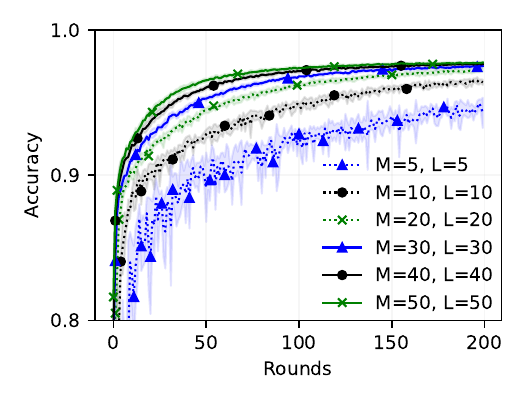}
	\caption{MNIST classification accuracy of FedRW in Dirichlet Non-IID ($\zeta_d$=0.1) setting with varying random walk chain counts and lengths.}
	\label{fig4}
\end{figure}

Fig. \ref{fig4} presents the classification accuracy of FedRW under different combinations of the number of RW chains $M$ and chain length $Q$. As $M$ and $Q$ increase, the model achieves higher accuracy, faster convergence, and improved stability. However, the performance gain tends to saturates when $M$, $Q \geq 20$, as the number of data samples visited per round becomes sufficient. When $M = Q =50$, the accuracy plateaus after 120 rounds, indicating that FedRW has converged.

Fig. \ref{fig_rw_avg_prox_nova} compares FedRW with FedAvg, FedProx \cite{11Li}, and FedNova \cite{Tackling2020Wang} in Dirichlet Non-IID settings. With proximal term coefficient tuned from $\{0.005, \dots, 0.1\}$, FedProx performs comparably to FedAvg, as its proximal term overly restricts local adaptation to extreme heterogeneity. FedNova slightly improves upon FedAvg but exhibits higher oscillations. This instability stems from its aggressive gradient scaling, which amplifies noise from clients with limited data in highly unbalanced data settings. In contrast, FedRW achieves the highest accuracy of 92.68\% and 61.75\% with superior stability. Its sequential traversal of local distributions acts as an implicit smoothing mechanism, effectively integrating diverse knowledge while avoiding the convergence volatility inherent in centralized aggregation.

\begin{figure}[!t]
	\centering	
		\subfloat[MNIST (Dirichlet Non-IID, $\zeta_d$=0.1)]{
			\includegraphics[width=0.48\linewidth]{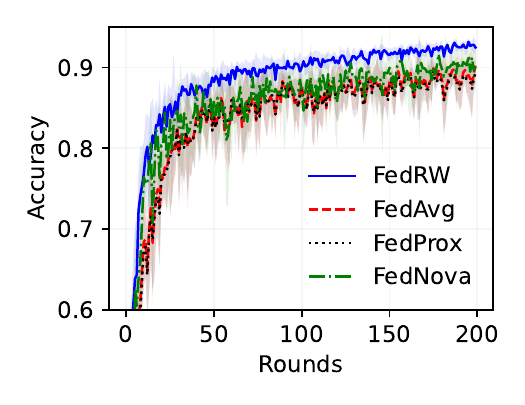}
		}
		\subfloat[CIFAR-100 (Dirichlet Non-IID, $\zeta_d$=0.2)]{
			\includegraphics[width=0.48\linewidth]{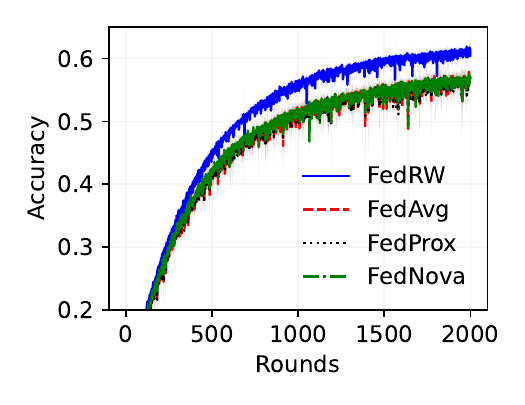}
		}
		\caption{Performance comparison between FedRW and state-of-the-art baselines in Dirichlet Non-IID setting.}
		\label{fig_rw_avg_prox_nova}
\end{figure}

\begin{figure}[!t]
	\centering	
		\subfloat[MNIST (mixed Non-IID, $\zeta_n$=0) \label{fig_mnist_gamma}]{
			\includegraphics[width=0.48\linewidth]{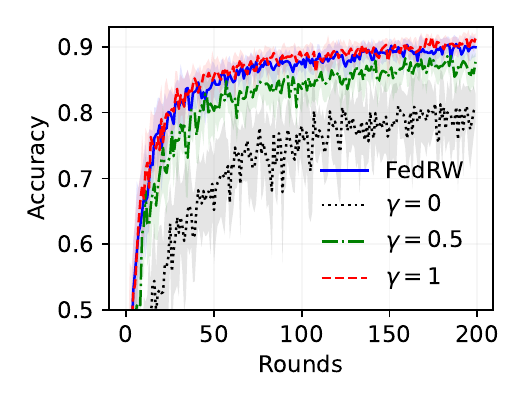}
		}
		\subfloat[CIFAR-100 (mixed Non-IID, $\zeta_n$=0.4) \label{fig_cifar100_gamma}]{
			\includegraphics[width=0.48\linewidth]{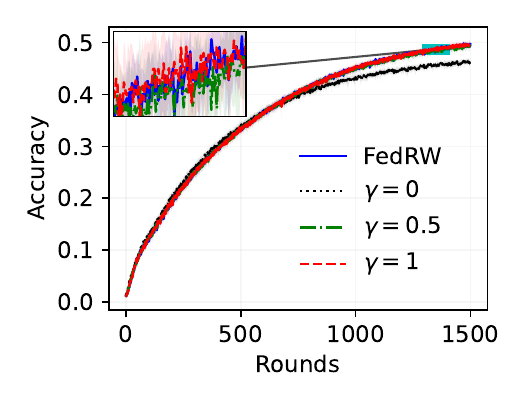}
		}
		\caption{Performance comparison of FedRW under varying degrees of intra-chain data diversity.}
		\label{fig_gamma}
\end{figure}

We introduce a parameter $\gamma$ to investigate intra-chain heterogeneity by controlling the data diversity of the $(1-\zeta_n)$ fraction within each chain. As shown in Fig. \ref{fig_gamma}, the highest accuracy and fastest convergence are achieved with fully heterogeneous chains ($\gamma=1$), with the standard FedRW random strategy performing closely. Conversely, homogeneous chains ($\gamma = 0$) suffer from gradient bias accumulation and sample homogenization, leading to significant performance degradation. The random strategy effectively enhances learning efficiency through implicit data mixing and mutual gradient correction. While the performance gap narrows under milder heterogeneity ($\zeta_n = 0.4$), the necessity of intra-chain diversity remains evident. These findings suggest that for extreme Non-IID scenarios, clustering clients by data similarity to systematically construct heterogeneous chains across clusters could further enhance learning efficiency, offering a promising direction for future optimization.

To validate the unweighted aggregation in (\ref{eq11}), we compared it against three alternative weighting schemes based on RW chain length (length randomly distributed in [3,10]), the number of unique clients visited per chain ($U=20$, $Q=10$), and the total data samples processed along the chain ($\zeta_d=0.1$) under scenarios favoring weighting. As shown in Fig. \ref{fig4_5}, unweighted aggregation achieves comparable testing loss while demonstrating better stability. Specifically, the oscillations observed in Fig. \ref{fig4_5} (c) for weighting by data volume stem from the fact that in the sequential structure of FedRW, a chain's final model already represents an iterative refinement of all prior participants. Assigning a large weight based on total data volume within a chain essentially treats this integrated model as a single monolithic node, which disproportionately amplifies the local biases of the final nodes in that chain. In contrast, unweighted aggregation treats each RW as an equal contribution, effectively smoothing out individual biases through path diversity and preserving a more balanced global representation. Furthermore, unweighted aggregation avoids the additional communication and storage overhead of tracking and transmitting chain metadata.

\begin{figure}[!t]
	\centering	
		\subfloat[MNIST ($\zeta_n$=0)]{
			\includegraphics[width=0.333\linewidth]{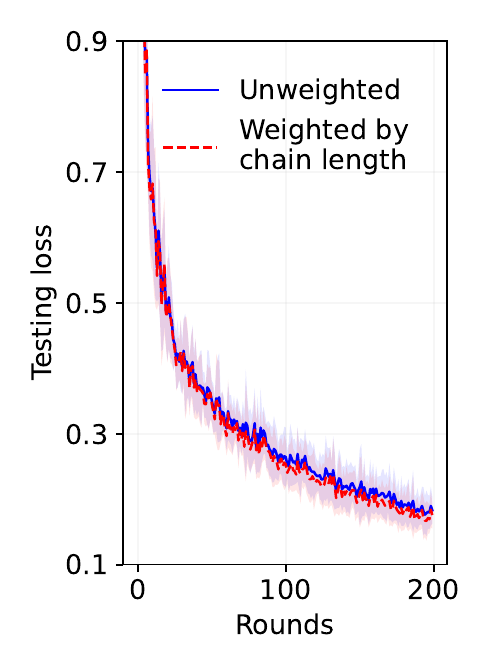}
		}
		\hspace{-0.54cm}  
		\subfloat[MNIST ($\zeta_n$=0)]{
			\includegraphics[width=0.333\linewidth]{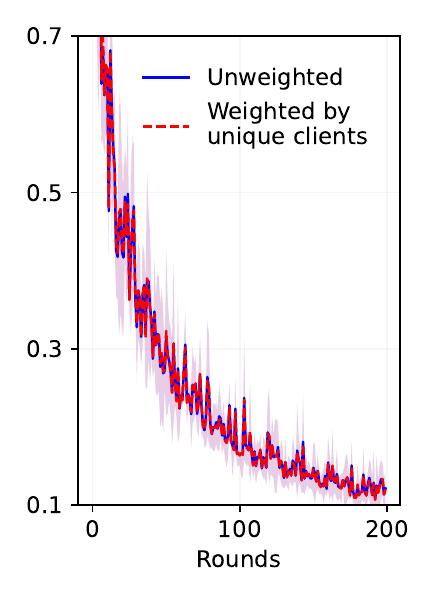}
		}
		\hspace{-0.54cm}  
		\subfloat[MNIST ($\zeta_d$=0.1)]{
			\includegraphics[width=0.333\linewidth]{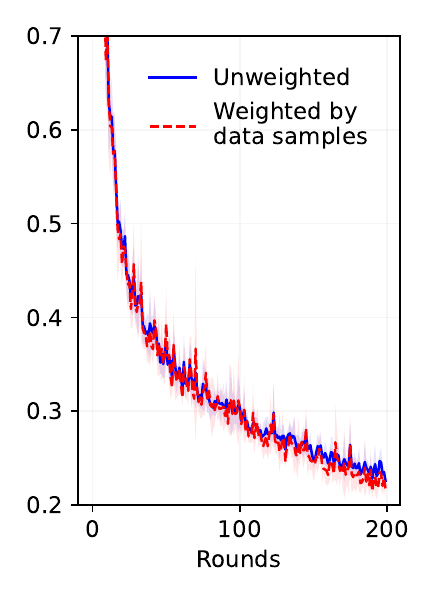}
		}
		\caption{Comparison between unweighted and weighted aggregation schemes.}
		\label{fig4_5}
\end{figure}

\subsection{Performance Comparisons of Joint RW Path Selection and Transmission Parameter Optimization Algorithm}
To the best of our knowledge, there is extremely limited prior work similar to FedRW, and none that addresses client selection or transmission parameter optimization under unreliable communication. To validate the proposed joint RW path selection and transmission parameter optimization algorithm, we compare it with three baselines: 1) Ideal FedRW: assumes perfect channels without transmission errors and employs a uniform random walk strategy; 2) RandParam FedRW: randomly selects transmission parameters but optimizes RW path selection based on the resulting link reliability. 3) Greedy-RW FedRW: optimizes transmission parameters while selecting RW paths greedily, where the next-hop node is selected as
\begin{equation}
	\label{eq38}
	i_{m,q+1}=\underset{j \in \mathcal{Z}_{m, q+1}}{\arg \max } E_{m, q}^j+\mathbb{I}_{q=Q} \cdot E_{m, Q}^j.
\end{equation}
where the indicator function $\mathbb{I}_{q = Q}$ activates only when $q = Q$, in which case the wireless link from the final client to the server is included. Note that as FedRW already outperforms all baselines in Section \ref{Section6}-A, we use only Ideal FedRW here. Our joint optimization nearly reaches this upper bound, other suboptimal curves are omitted for clarity.

\begin{figure}[!t]
	\centering	
		\subfloat[MNIST (mixed Non-IID, $\zeta_n$=0)]{
			\includegraphics[width=0.48\linewidth]{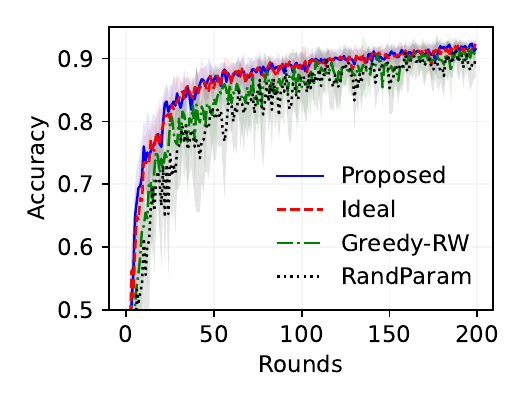}
		}
		\subfloat[MNIST (mixed Non-IID, $\zeta_n$=0)]{
			\includegraphics[width=0.48\linewidth]{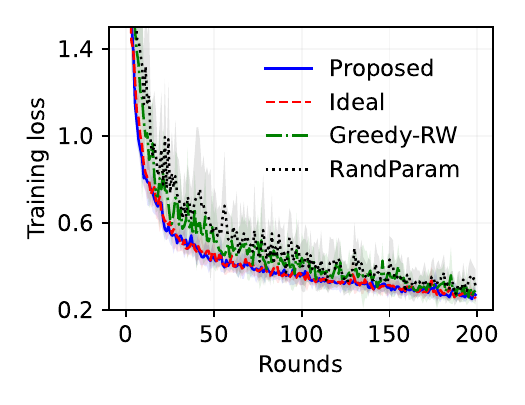}
	}
	
	\vspace{-0.4cm} 
	
		\subfloat[Fashion-MNIST (mixed Non-IID, $\zeta_n$=0)]{
			\includegraphics[width=0.48\linewidth]{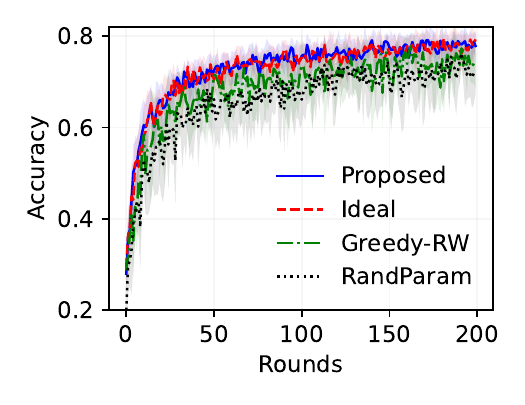}
		}
		\subfloat[Fashion-MNIST (mixed Non-IID, $\zeta_n$=0)]{
			\includegraphics[width=0.48\linewidth]{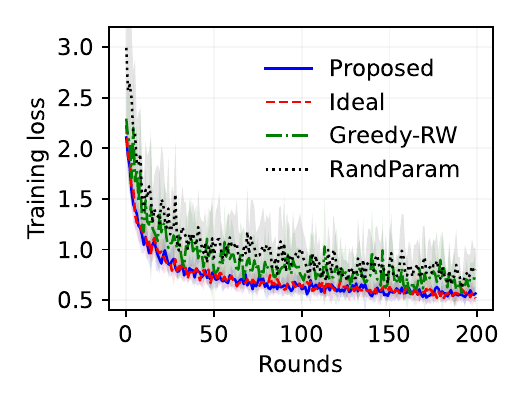}
	}
	
	\vspace{-0.4cm} 
	
		\subfloat[CIFAR-10 (mixed Non-IID, $\zeta_n$=0.2)]{
			\includegraphics[width=0.48\linewidth]{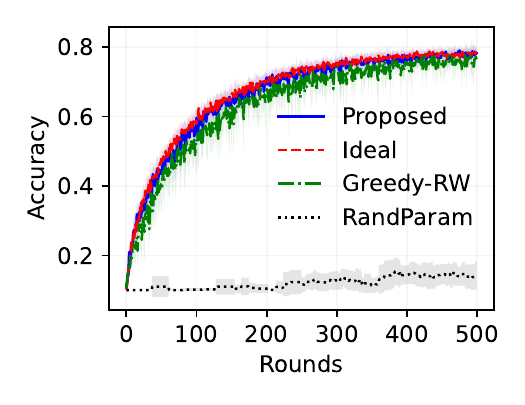}
		}
		\subfloat[CIFAR-10 (mixed Non-IID, $\zeta_n$=0.2)]{
			\includegraphics[width=0.48\linewidth]{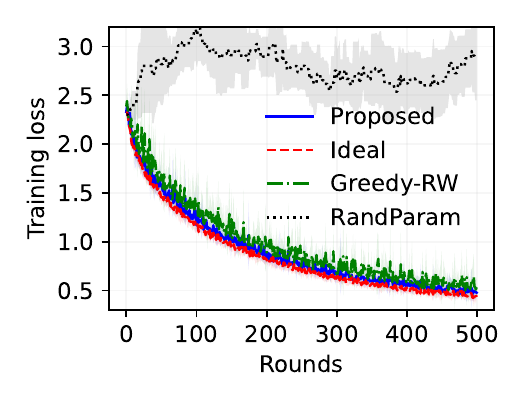}
	}

    \vspace{-0.4cm} 
	
		\subfloat[CIFAR-100 (mixed Non-IID, $\zeta_n$=0.4)]{
			\includegraphics[width=0.48\linewidth]{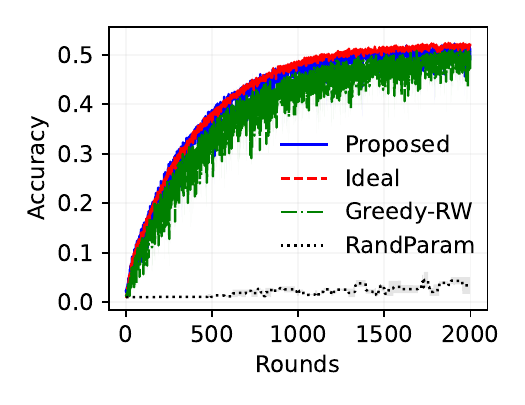}
		}
		\subfloat[CIFAR-100 (mixed Non-IID, $\zeta_n$=0.4)]{
			\includegraphics[width=0.48\linewidth]{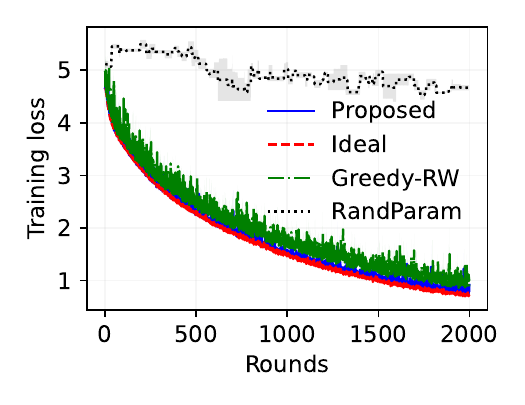}
		}
		\caption{Comparative of classification accuracy and loss between the proposed joint RW path selection and transmission parameter optimization algorithm and baselines on MNIST, Fashion-MNIST, CIFAR-10, and CIFAR-100.}
		\label{fig5}
\end{figure}

Fig. \ref{fig5} (a) and (b) show the accuracy and training loss of proposed algorithm and baselines over training rounds on the MNIST classification. Compared to Greedy-RW and RandParam algorithms, the proposed algorithm achieves an average accuracy improvement of 2.78\% and 4.67\% and has a faster convergence rate. The 2.78\% performance gain is attributed to the foresight of the dynamic pruning-based beam search algorithm, which avoids the short-sighted decisions typical of greedy algorithms. This enables the proposed algorithm to match the accuracy and convergence rate of the Ideal FedRW. The 4.67\% improvement comes from the optimization of transmission parameters, which allows access to more and higher quality neighbors during RW path selection. In contrast, Greedy-RW and RandParam algorithms experience larger oscillations due to discarding models with transmission errors. Additionally, the training loss suggests that the empirical results are consistent with the expected convergence behavior outlined in the theoretical analysis. Fig. \ref{fig5} (c) $-$ (h) present the performance on Fashion-MNIST, CIFAR-10, and CIFAR-100. Consistent with the findings on MNIST, the proposed algorithm achieves the highest accuracy with minimal convergence oscillation across all three datasets. For CIFAR-10 and CIFAR-100 in particular, RandParam suffers from ineffective training, as its random transmission parameters often cause most wireless links to violate the relatively strict delay constraints. These results underscore the necessity of jointly optimizing RW path selection and transmission parameters to reduce training loss over unreliable wireless networks.

From Fig. \ref{fig5}, we observe that our proposed algorithm achieves stable convergence and improved performance on non-convex neural networks. This empirical evidence suggests that, although these models are non-convex, their loss landscapes during training may possess properties akin to the P{\L} condition assumed in our theoretical analysis, making our convergence guarantees practically relevant.

\begin{figure}[!t]
	\centering	
		\subfloat[Fashion-MNIST (mixed Non-IID, $\zeta_n$=0)]{
			\includegraphics[width=0.48\linewidth]{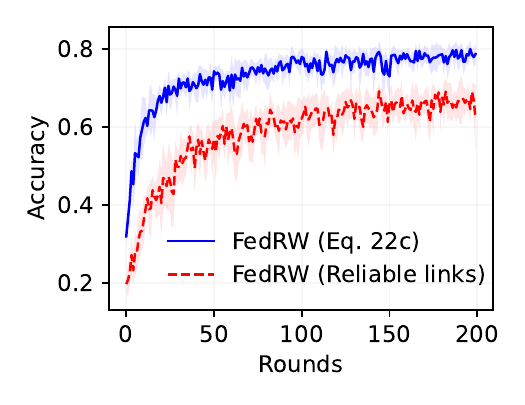}
		}
		\subfloat[Fashion-MNIST (Dirichlet Non-IID, $\zeta_d$=0.1)]{
			\includegraphics[width=0.48\linewidth]{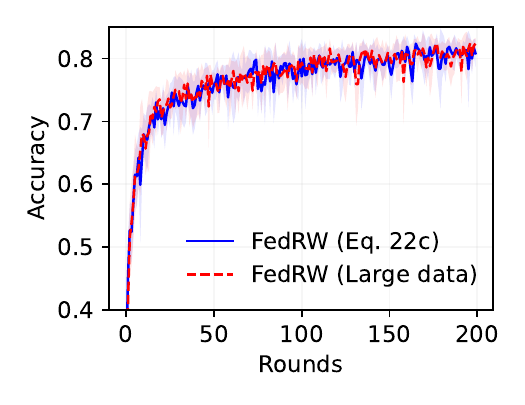}
		}
		\caption{Comparison of classification accuracy between the proposed algorithm and its variant with constraint (\ref{eq27c}) relaxed.}
		\label{fig_relax}
\end{figure}

We analyze the impact of relaxing constraint (\ref{eq27c}) on convergence to verify that this restriction is essential for heterogeneity mitigation. As shown in Fig. \ref{fig_relax} (a), prioritizing reliable links by allowing clients with high wireless reliability to participate in multiple chains leads to a final accuracy drop of 11.3\%. This occurs because the optimization objective in (\ref{eq27}) favors high reliability, causing the algorithm to persistently select the same few optimal clients and creating severe data sampling bias. As shown in Fig. \ref{fig_relax} (b), allowing the top 10\% of clients with the largest data volume to participate in multiple chains yields no advantage over standard FedRW. The inherent bias of beam search toward reliable links dominates data size considerations, while repeated selection of the same Non-IID clients introduces redundancy without enhancing data diversity. Therefore, constraint (\ref{eq27c}) ensures that FedRW captures a broader range of data distributions across the network, effectively preventing the global model from being dominated by a small subset of biased clients.

\begin{figure*}[!t]
	\centering	
	\subfloat[$\gamma_R$ vs. Number of aggregated RW chains]{
		\includegraphics[width=0.27\linewidth]{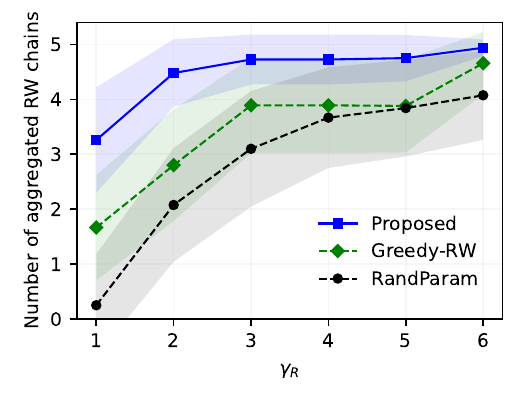}
		}
	\subfloat[$\gamma_\tau$ vs. Number of aggregated RW chains]{
		\includegraphics[width=0.27\linewidth]{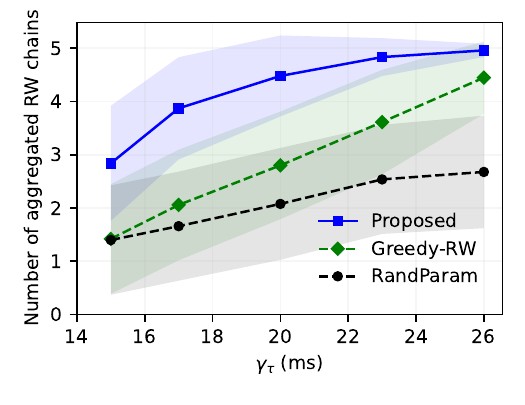}
	}
	\subfloat[BER vs. Number of aggregated RW chains]{
		\includegraphics[width=0.27\linewidth]{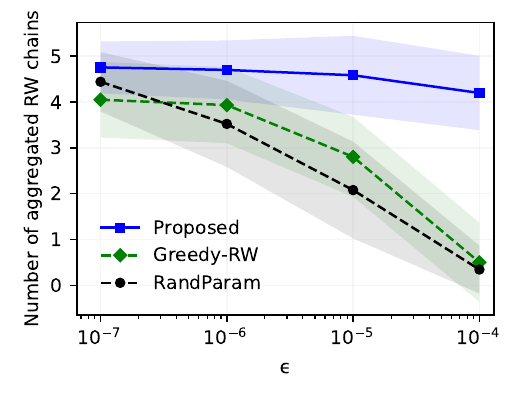}
	}
	\caption{Comparison of number of aggregated RW chains under different network conditions on MNIST in mixed Non-IID ($\zeta_n$=0) setting.}
	\label{fig7}
\end{figure*}

We further evaluate the proposed algorithm's advantages over baselines under varying network conditions. Fig. \ref{fig7} (a) shows how the number of RW chains participating in aggregation changes with $\gamma_R$. The proposed algorithm consistently achieves the best performance across all $\gamma_R$. As $\gamma_R$ increases, the success probability of model parameter transmission improves exponentially, narrowing the performance gap among algorithms. When $\gamma_R = 6$, nearly all wireless links with optimized transmission parameters become reliable, allowing Greedy-RW to perform comparably to the proposed algorithm.

Fig. \ref{fig7} (b) presents the number of RW chains aggregated under varying packet delay constraints $\gamma_\tau$. As the delay constraint is relaxed, smaller packet sizes or more retransmissions are allowed, which improves transmission success rates under unreliable channels. The proposed algorithm consistently outperforms all baselines across all levels of delay constraints. Greedy-RW grows faster than RandParam because optimized transmission parameters allow it to better leverage the available delay budget to improve transmission reliability.

Fig. \ref{fig7} (c) shows the number of aggregated RW chains under varying BER levels. As BER increases, all methods degrade due to reduced transmission reliability, but the performance gap widens as baselines deteriorate more rapidly. RandParam fails to adjust transmission parameters, making it hard to find usable next-hop links. Greedy-RW can find reliable links locally, but often cannot extend the path due to surrounding unreliable neighbors. Compared to the baselines, the proposed algorithm effectively addresses both limitations by jointly optimizing transmission parameters and path selection, resulting in robust aggregation even under high BER.

\section{Conclusion}
\label{Section7}
In this paper, we address two key challenges in wireless federated learning: data heterogeneity and unreliable transmission. We first propose a new FedRW framework, where multiple RW chains perform local updates in parallel, and their updates are aggregated by a server to mitigate data heterogeneity. Building on this, we formulate a joint RW path selection and transmission parameter optimization problem, aiming to minimize the training loss of FedRW. We simplify the problem using an upper bound of the expected convergence over unreliable wireless networks and develop a distributed solution. Each server or client only needs to determine the optimal packet size, maximum number of retransmissions, and next-hop based on the wireless link reliability and latency to its neighbors. Simulation results demonstrate the effectiveness and superiority of the proposed wireless FedRW in handling data heterogeneity and reducing training loss.

While our analysis provides convergence upper bounds, two theoretical questions remain open. First, matching lower bounds could be established by applying rate distortion theory to model gradient transmission over unreliable links, or through adversarial constructions with orthogonal client data. Second, extending the analysis beyond P{\L} conditions requires new techniques, such as constructing Lyapunov functions that incorporate wireless conditions, data heterogeneity, and random walk dynamics.

\begin{appendices}
\section{} \label{appendixA}
According to (\ref{eq19}), the global model update at round $t$ is given by
\begin{equation}
	\label{eq39}
	\mathbf{w}^{t+1}=\mathbf{w}^t-\lambda\left(\nabla F\left(\mathbf{w}^t\right)-\Theta\right),
\end{equation}
where \\
$\Theta=\nabla F\left(\mathbf{w}^t\right)-\frac{\sum_{m=1}^M \prod_{q=0}^Q \sum_{i \in \mathcal{N}_{m, q}} a_{m, q}^i C_{m, q}^i \nabla f_{i_{m,Q}}\left(\mathbf{w}_{m,Q-1}^t\right)}{\sum_{m=1}^M \prod_{q=0}^Q \sum_{i \in \mathcal{N}_{m, q}} a_{m, q}^i C_{m, q}^i}$.

To prove Theorem \ref{theorem1}, we leverage the quadratic upper bound provided by $L$-smoothness.
\begin{equation}
	\label{eq40}
	\begin{aligned}
		& F\left(\mathbf{w}^{t+1}\right)\leq F\left(\mathbf{w}^t\right)+\left(\mathbf{w}^{t+1}-\mathbf{w}^t\right)^\top \nabla F\left(\mathbf{w}^t\right)\\
		& \qquad\quad\qquad+\frac{L}{2}\left\|\mathbf{w}^{t+1}-\mathbf{w}^t\right\|^2.
	\end{aligned}
\end{equation}

Given the learning rate $\lambda = \frac{1}{L}$, we can express the expected loss as
\begin{equation}
	\label{eq41}
    \begin{aligned}
	    \mathbb{E}\left[F\left(\mathbf{w}^{t+1}\right)\right]& \leq \mathbb{E}\bigg[F\left(\mathbf{w}^t\right)-\lambda\left(\nabla F\left(\mathbf{w}^t\right)-\Theta\right)^\top \nabla F\left(\mathbf{w}^t\right) \\
	    & \quad +\frac{\left\|\nabla F\left(\mathbf{w}^t\right)\right\|^2}{2 L}-\frac{\Theta^\top \nabla F\left(\mathbf{w}^t\right)}{L}+\frac{\|\Theta\|^2}{2 L}\bigg] \\
	    & = \mathbb{E}\left[F\left(\mathbf{w}^t\right)\right]-\frac{\left\|\nabla F\left(\mathbf{w}^t\right)\right\|^2}{2 L}+\frac{\mathbb{E}\left[\|\Theta\|^2\right]}{2 L}.
    \end{aligned}
\end{equation}

Then, we derive the upper bound of $\mathbb{E}[\|\Theta\|^2]$. Let $\nabla \mathbf{g}_{i_{m,Q}}^t = \nabla f_{i_{m,Q}}(\mathbf{w}_{m,Q-1}^t)$, we first decompose (\ref{eq18}) as $\nabla F(\mathbf{w}^t)=\frac{1}{\psi}(\sum_{m \in \Psi_{s}} \nabla \mathbf{g}_{i_{m,Q}}^t+\sum_{m \in \Psi_f} \nabla \mathbf{g}_{i_{m,Q}}^t)$, where $\Psi_s$ denotes the set of RW chains that are selected and successfully transmitted in round $t$, and $\Psi_f$ denotes those that are either not selected or failed in transmission. Since $\sum_{m=1}^M \prod_{q=0}^Q \sum_{i \in \mathcal{N}_{m, q}} a_{m, q}^i C_{m, q}^i \nabla \mathbf{g}_{i_{m,Q}}^t=\sum_{m \in \Psi_s} \nabla \mathbf{g}_{i_{m,Q}}^t$. Therefore, we have
\begin{equation}
	\label{eq42}
	\begin{aligned}
		\mathbb{E}\left[\|\Theta\|^2\right]&=\mathbb{E}\Bigg[\bigg\|\frac{1}{\psi}\bigg[\sum_{m \in \Psi_s} \nabla \mathbf{g}_{i_{m,Q}}^t+\sum_{m \in \Psi_f} \nabla \mathbf{g}_{i_{m,Q}}^t\bigg] \\
		&\quad -\frac{\sum_{m \in \Psi_s} \nabla \mathbf{g}_{i_{m,Q}}^t}{\sum_{m=1}^M \prod_{q=0}^Q \sum_{i \in \mathcal{N}_{m, q}} a_{m, q}^i C_{m, q}^i}\bigg\|^2\Bigg] \\
		& =\mathbb{E}\Bigg[\bigg\|-\frac{\psi-\sum_{m=1}^M \prod_{q=0}^Q \sum_{i \in \mathcal{N}_{m, q}} a_{m, q}^i C_{m, q}^i}{\psi \sum_{m=1}^M \prod_{q=0}^Q \sum_{i \in \mathcal{N}_{m, q}} a_{m, q}^i C_{m, q}^i}  \\
		& \quad \times \sum_{m \in \Psi_s} \nabla \mathbf{g}_{i_{m,Q}}^t +\frac{1}{\psi} \sum_{m \in \Psi_f} \nabla \mathbf{g}_{i_{m,Q}}^t \bigg\|^2\Bigg] .
	\end{aligned}
\end{equation}

By applying the triangle inequality, we have
\begin{equation}
	\label{eq43}
	\begin{aligned}
		&\mathbb{E}\left[\|\Theta\|^2\right] \leq  \mathbb{E}\Bigg[\frac{\psi-\sum_{m=1}^M \prod_{q=0}^Q \sum_{i \in \mathcal{N}_{m, q}} a_{m, q}^i C_{m, q}^i}{\psi\sum_{m=1}^M \prod_{q=0}^Q \sum_{i \in \mathcal{N}_{m, q}} a_{m, q}^i C_{m, q}^i}  \\
		& \qquad \quad \times \sum_{m \in \Psi_s}\left\|\nabla \mathbf{g}_{i_{m,Q}}^t\right\| +\frac{1}{\psi} \sum_{m \in \Psi_f}\left\|\nabla \mathbf{g}_{i_{m,Q}}^t\right\|\Bigg]^2 .
	\end{aligned}
\end{equation}

Using $\|\nabla \mathbf{g}_{i_{m,Q}}^t\| \leq \sqrt{\alpha^2+\sigma^2\|\nabla F(\mathbf{w}^t)\|^2}$ from Definition \ref{definition1}, we have
\begin{equation}
	\label{eq44}
	\begin{aligned}
	    & \sum_{m \in \Psi_s}\left\|\nabla \mathbf{g}_{i_{m,Q}}^t\right\|\\
	    & \quad \leq \sqrt{\alpha^2+\sigma^2\left\|\nabla F\left(\mathbf{w}^t\right)\right\|^2} \sum_{m=1}^M \prod_{q=0}^Q \sum_{i \in \mathcal{N}_{m, q}} a_{m, q}^i C_{m, q}^i.
    \end{aligned}
\end{equation}
Noting that $|\Psi_s| + |\Psi_f| = \Psi$, we also have $\sum_{m \in \Psi_f}\|\nabla \mathbf{g}_{i_{m,Q}}^t\| \leq \sqrt{\alpha^2+\sigma^2\|\nabla F(\mathbf{w}^t)\|^2}(\psi-\sum_{m=1}^M \prod_{q=0}^Q \sum_{i \in \mathcal{N}_{m, q}} a_{m, q}^i C_{m, q}^i)$. Therefore, (\ref{eq43}) can be expressed as
\begin{equation}
	\label{eq45}
	\begin{aligned}
		\mathbb{E}\left[\|\Theta\|^2\right] &\leq \mathbb{E}\Bigg[\frac{2}{\psi}\bigg(\psi-\sum_{m=1}^M \prod_{q=0}^Q \sum_{i \in \mathcal{N}_{m, q}} a_{m, q}^i C_{m, q}^i\bigg)\\
		& \quad \times \sqrt{\alpha^2+\sigma^2\left\|\nabla F\left(\mathbf{w}^t\right)\right\|^2}\Bigg]^2 \\
		& =\frac{4}{\psi^2} \mathbb{E}\Bigg[\psi-\sum_{m=1}^M \prod_{q=0}^Q \sum_{i \in \mathcal{N}_{m, q}} a_{m, q}^i C_{m, q}^i\Bigg]^2\\
		&\quad \times \left(\alpha^2+\sigma^2\left\|\nabla F\left(\mathbf{w}^t\right)\right\|^2\right).
	\end{aligned}
\end{equation}

Based on the expectation of the transmission success indicator $\mathbb{E}[C_{m, q}^i] = (p_{m, q}^i)^{\beta_{m, q}^i} = E_{m, q}^i$, and $\sum_{m \in \Psi_s} 1+\sum_{m \in \Psi_f} 1=\psi$, we have
\begin{equation}
	\label{eq46}
	\begin{aligned}
		\mathbb{E}\left[\|\Theta\|^2\right] \leq&  \frac{4}{\psi}\bigg(\psi-\sum_{m=1}^M \prod_{q=0}^Q \sum_{i \in \mathcal{N}_{m, q}} a_{m, q}^i E_{m, q}^i\bigg) \\
		& \times \left(\alpha^2+\sigma^2\left\|\nabla F\left(\mathbf{w}^t\right)\right\|^2\right),
	\end{aligned}
\end{equation}

Substituting (\ref{eq46}) into (\ref{eq41}) and subtracting $\mathbb{E}[F(\mathbf{w}^*)]$ in both sides, we obtain
\begin{equation}
	\label{eq47}
	\begin{aligned}
		& \mathbb{E}\left[F\left(\mathbf{w}^{t+1}\right)-F\left(\mathbf{w}^*\right)\right] \leq \mathbb{E}\left[F\left(\mathbf{w}^t\right)-F\left(\mathbf{w}^*\right)\right]   \\
		& \quad -\frac{\left\|\nabla F\left(\mathbf{w}^t\right)\right\|^2}{2 L}+\frac{2}{\psi L}\bigg(\psi-\sum_{m=1}^M \prod_{q=0}^Q \sum_{i \in \mathcal{N}_{m, q}} a_{m, q}^i E_{m, q}^i\bigg) \\
		& \quad \times \left(\alpha^2+\sigma^2\left\|\nabla F\left(\mathbf{w}^t\right)\right\|^2\right) \\
		& =\mathbb{E}\left[F\left(\mathbf{w}^t\right)-F\left(\mathbf{w}^*\right)\right] -\frac{\left\|\nabla F\left(\mathbf{w}^t\right)\right\|^2}{2 L}\\
		& \quad +\frac{2 \alpha^2}{\psi L}\bigg(\psi-\sum_{m=1}^M \prod_{q=0}^Q \sum_{i \in \mathcal{N}_{m, q}} a_{m, q}^i E_{m, q}^i \bigg) \\
		& \quad + \frac{2 \sigma^2\left\|\nabla F\left(\mathbf{w}^t\right)\right\|^2}{\psi L}\bigg(\psi-\sum_{m=1}^M \prod_{q=0}^Q \sum_{i \in \mathcal{N}_{m, q}} a_{m, q}^i E_{m, q}^i\bigg).
	\end{aligned}
\end{equation}

Since $F$ satisfies the $\mu$-P{\L} property, we have
\begin{equation}
	\label{eq48}
	\left\|\nabla F\left(\mathbf{w}^t\right)\right\|^2 \geq 2 \mu\left(F\left(\mathbf{w}^t\right)-F\left(\mathbf{w}^*\right)\right).
\end{equation}

Substituting (\ref{eq48}) into (\ref{eq47}), we have
\begin{equation}
	\label{eq49}
	\begin{aligned}
		& \mathbb{E}\left[F\left(\mathbf{w}^{t+1}\right)-F\left(\mathbf{w}^*\right)\right] \leq \mathcal{J} \mathbb{E}\left[F\left(\mathbf{w}^t\right)-F\left(\mathbf{w}^*\right)\right] \\
		& +\frac{2 \alpha^2}{\psi L}\bigg(\psi-\sum_{m=1}^M \prod_{q=0}^Q \sum_{i \in \mathcal{N}_{m, q}} a_{m, q}^i E_{m, q}^i\bigg),
	\end{aligned}
\end{equation}
where\\
$\mathcal{J}=1-\frac{\mu}{L}+\frac{4 \sigma^2 \mu}{\psi L} (\psi-\sum_{m=1}^M \prod_{q=0}^Q \sum_{i \in \mathcal{N}_{m, q}} a_{m, q}^i E_{m, q}^i)$.

By recursive (\ref{eq49}), we have
\begin{equation}
	\label{eq50}
	\begin{aligned}
		& \mathbb{E}\left[F\left(\mathbf{w}^{t+1}\right)-F\left(\mathbf{w}^*\right)\right] \leq \mathcal{J}^t \mathbb{E}\left[F\left(\mathbf{w}^0\right)-F\left(\mathbf{w}^*\right)\right] \\
		& +\frac{2 \alpha^2}{\psi L}\bigg(\psi-\sum_{m=1}^M \prod_{q=0}^Q \sum_{i \in \mathcal{N}_{m, q}} a_{m, q}^i E_{m, q}^i\bigg) \frac{1-\mathcal{J}^t}{1-\mathcal{J}}.
	\end{aligned}
\end{equation}
This completes the proof.
\end{appendices}

\bibliographystyle{IEEEtran}
\bibliography{IEEEabrv,ref}

\end{document}